\documentclass{article}

\usepackage[preprint]{neurips_2025}

\usepackage[utf8]{inputenc} % allow utf-8 input
\usepackage[T1]{fontenc}    % use 8-bit T1 fonts
\usepackage{hyperref}       % hyperlinks
\usepackage{url}            % simple URL typesetting
\usepackage{booktabs}       % professional-quality tables
\usepackage{amsfonts}       % blackboard math symbols
\usepackage{amsthm}         % theorem environments, including remark
\usepackage{nicefrac}       % compact symbols for 1/2, etc.
\usepackage{microtype}      % microtypography
\usepackage{xcolor}         % colors

\usepackage{multirow} % 用于合并多行单元格
\usepackage{graphicx} % 用于缩小表格
\usepackage{array} % 用于调整列宽
\usepackage{wrapfig} 
\usepackage{bm}
\usepackage{enumitem}
\usepackage{amssymb}
\usepackage{amsmath}% 加载 amsmath 宏包，用于数学公式
\usepackage{amsthm}% 加载 amsthm 宏包，用于定义定理环境
\usepackage{algorithm}
\usepackage{algorithmic}
\usepackage{ulem}
\newtheorem{theorem}{Theorem}[section]% 定理编号按章节编号
\newtheorem{lemma}[theorem]{Lemma}% 引理与定理共享编号
\newtheorem{remark}[theorem]{Remark}

\title{Beyond Personalization: Federated Recommendation with Calibration via Low-rank Decomposition}

\author{\hspace{-0.6cm}\textbf{Jundong Chen}\textsuperscript{1,2},\textbf{Honglei Zhang}\textsuperscript{1,2},\textbf{Haoxuan Li}\textsuperscript{3},\textbf{Chunxu Zhang}\textsuperscript{4},\textbf{Zhiwei Li}\textsuperscript{5},\textbf{Yidong Li}\textsuperscript{1,2}\thanks{Corresponding author.} \\
\hspace{-0.5cm}\textsuperscript{1}Key Laboratory of Big Data \& Artificial Intelligence in Transportation, Ministry of Education,\\
\hspace{-0.4cm}\textsuperscript{2}Beijing Jiaotong University, 
\textsuperscript{3}Peking University, \textsuperscript{4}Jilin University, \textsuperscript{5}University of Technology Sydney\\
\hspace{-0.4cm}\texttt{jundongchen@bjtu.edu.cn},\hspace{0.2cm}\texttt{ydli@bjtu.edu.cn}
}
\begin{document}

\maketitle

\begin{abstract}
  Federated recommendation (FR) is a promising paradigm to protect user privacy in recommender systems. Distinct from general federated scenarios, FR inherently needs to preserve client-specific parameters, \textit{i.e.}, user embeddings, for privacy and personalization. However, we empirically find that globally aggregated item embeddings can induce skew in user embeddings, resulting in suboptimal performance. To this end, we theoretically analyze the \textit{user embedding skew} issue and propose \textbf{P}ersonalized \textbf{Fed}erated recommendation with \textbf{C}alibration via \textbf{L}ow-\textbf{R}ank decomposition (\textbf{PFedCLR}). Specifically, PFedCLR introduces an integrated dual-function mechanism, implemented with a buffer matrix, to jointly calibrate local user embedding and personalize global item embeddings. To ensure efficiency, we employ a low-rank decomposition of the buffer matrix to reduce the model overhead. Furthermore, for privacy, we train and upload the local model before personalization, preventing the server from accessing sensitive information. Extensive experiments demonstrate that PFedCLR effectively mitigates user embedding skew and achieves a desirable trade-off among performance, efficiency, and privacy, outperforming state-of-the-art (SOTA) methods. Our code is available at \url{https://github.com/jundongchen13/PFedCLR}.
  % Our code is available at \url{https://anonymous.4open.science/r/PFedCLR}.
  % at XXXXXXXXXXXXXXXXXXXX.
\end{abstract}

\section{Introduction}\label{sec:intro}
Federated recommendation (FR) has emerged as a new architecture to provide customized recommendation services while preserving user privacy~\cite{sun2024survey, yin2024device, zhang2025personalized}. Unlike general federated learning (FL), FR inherently requires client-specific parameters to ensure personalization and privacy~\cite{ammad2019federated,chai2020secure}. Specifically, typical FR methods, such as FedMF~\cite{chai2020secure} and FedNCF~\cite{perifanis2022federated}, treat item embeddings as the global model while keeping user embedding as the local variable, as illustrated in Figure~\ref{pic:method_comparison} (a). However, such FR methods of assigning one same global model to all clients are suboptimal, since there exists natural heterogeneity across clients, \textit{e.g.}, non-independent and identical distribution (Non-IID) of their interaction data~\cite{sun2024survey, zhang2023lightfr}. Hence, personalized federated recommendation (pFR) has been developed to tailor the personalized model for each client~\cite{zhang2023dual, jiang2024tutorial}. 

Conceptually, pFR incorporates the personalization mechanism to derive a personalized global model for each client, as illustrated in Figure~\ref{pic:method_comparison} (b). Existing personalization mechanisms can be categorized into two types~\cite{zheng2025confree}: server-side global aggregation and client-side local adaptation. From the server, pFedGraph~\cite{ye2023personalized} optimizes a unique global model for each client based on similarity, and GPFedRec~\cite{zhang2024gpfedrec} proposes graph-guided aggregation to yield user-specific item embeddings. From the client, PFedRec~\cite{zhang2023dual} allows dual personalization for score function and item embeddings, and FedRAP~\cite{li2023federated} adapts global model by applying an additive local model. These methods all personalize the global item embeddings, effectively enhancing the customized recommendations for each client.

\begin{figure*}[t]
\centering
\includegraphics[width=0.8\textwidth]{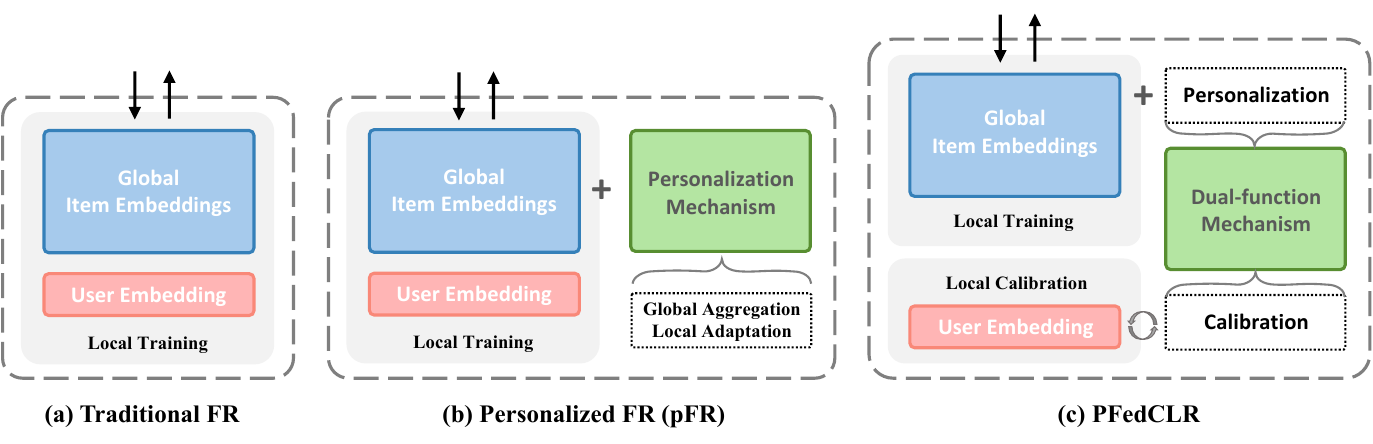} % Reduce the figure size so that it is slightly narrower than the column.
\caption{Different frameworks for federated recommendation (FR). While existing pFR methods mainly focus on the personalization of global item embeddings, PFedCLR additionally achieves the calibration of local user embedding.}
\label{pic:method_comparison}
\end{figure*}

\begin{wrapfigure}{r}{0.54\textwidth}
% \vspace{-0.4cm}
\centering
\includegraphics[width=0.54\textwidth]{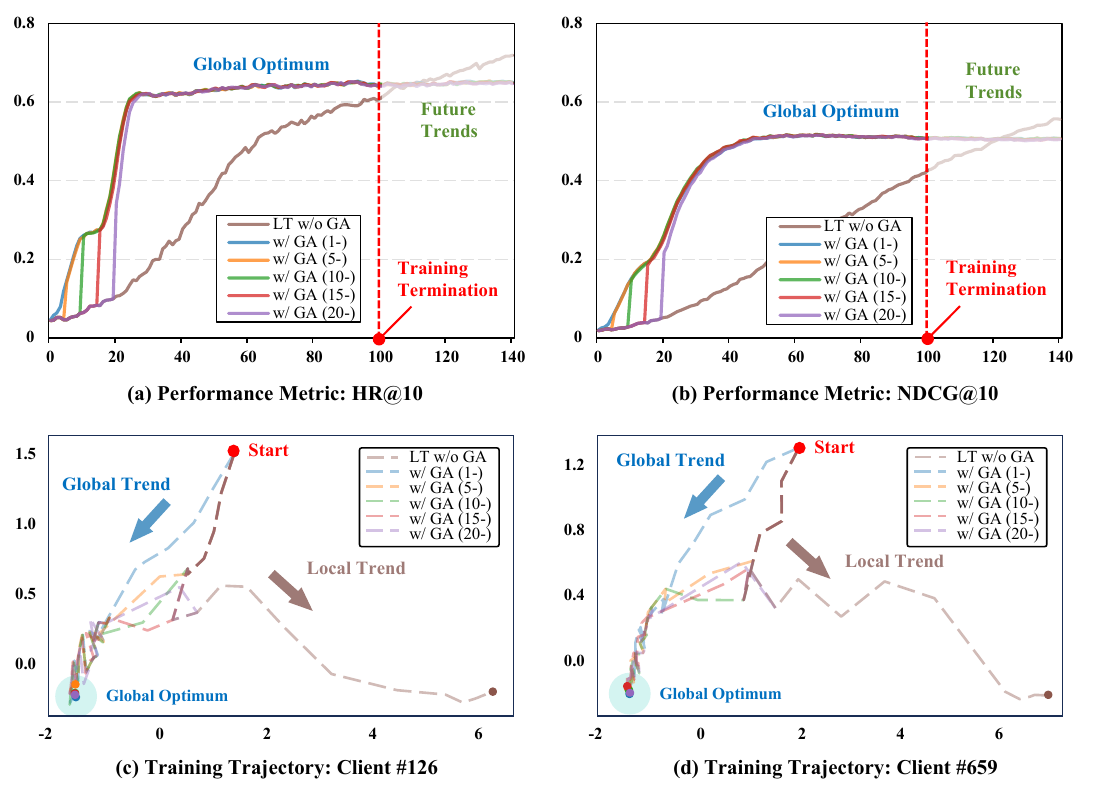} % Model performance and t-SNE training trajectory of local training with and without global aggregation.
\vspace{-0.4cm}
\caption{\textbf{Motivation.} "LT w/o GA" denotes only local training without global aggregation, while "w/ GA ($i$-)" denotes that global aggregation is performed from round $i$ onwards, rather than local training only.}
\vspace{-0.2cm}
\label{pic:pre_exp}
\end{wrapfigure}

However, we empirically find that the globally aggregated item embeddings can distort the training trajectories of local user embedding, leading to suboptimal solutions. In the experiment, we start global aggregation at different global rounds and set the total number of global rounds to 100, following the mainstream methods~\cite{chai2020secure,perifanis2022federated,zhang2023dual,li2023federated}. Additionally, 40 extra rounds are conducted to observe long-term trends. As the performance shown in Figure~\ref{pic:pre_exp} (a)(b), we find that: I) Under the standard setting of 100 global rounds, the groups with global aggregation converge to the same global optimum\footnote{In this paper, "global optimum" and "local optimum" refer to collaborative optimum under global aggregation and individual optimum under local training, respectively, rather than the concepts in optimization theory.}, achieving better performance than the one with local training only. II) From future trends, the group without global aggregation can achieve superior performance. Considering the specificity of FR scenarios, we further visualize the t-SNE training trajectories of local user embeddings. As shown by the two examples in Figure~\ref{pic:pre_exp} (c)(d), the user embeddings follow a consistent updating trend after starting global aggregation, which differs significantly from the only local training trend. More experimental details can be found in Appendix~\ref{app:pre_exp}. Based on the observations above, we argue that the item embeddings after global aggregation affect the optimization direction of the local user embedding, potentially leading to suboptimal performance, \textit{i.e.}, the \textit{user embedding skew} issue. Neither traditional FR nor pFR methods take this into account.
% In fact, user embedding, as the local variable, should be updated by local information, instead of the global information introduced by aggregation. 

To address this issue, we propose a calibrated pFR method called \textbf{P}ersonalized \textbf{Fed}erated recommendation with \textbf{C}alibration via \textbf{L}ow-\textbf{R}ank decomposition (\textbf{PFedCLR}), as shown in Figure~\ref{pic:method_comparison} (c). Beyond the personalization mechanism of mainstream pFR methods, PFedCLR incorporates a calibration mechanism to mitigate the effect of global information from aggregation on local user embeddings. Specifically, our key contributions are as follows:
\begin{itemize}[left=0pt, itemsep=-1pt]
    \item To the best of our knowledge, we first identify the \textit{user embedding skew} issue specific to FR, and provide an analysis from both experimental and theoretical perspectives.
    \item We propose an integrated dual-function mechanism to achieve both the calibration of local user embedding and the personalization of global item embeddings. Implemented via a buffer matrix, this mechanism effectively mitigates the previously overlooked \textit{user embedding skew} issue.
    \item For efficiency, PFedCLR employs a low-rank decomposition of the buffer matrix to achieve a lightweight and efficient design. For privacy, the local model is uploaded before personalization, preventing the server from accessing sensitive user information.
    % Moreover, in our method, the server and clients can compute synchronously, which enhances training efficiency.
    % unlike other FR methods, where global aggregation and local training are performed asynchronously,
    \item Through extensive experiments on five real-world datasets, we demonstrate that PFedCLR outperforms other state-of-the-art (SOTA) methods, while ensuring both efficiency and privacy.
\end{itemize}

\section{Related Work}
\subsection{Personalized federated learning}
% Building upon this, pFedMe~\cite{t2020personalized} introduces the Moreau envelope for further improvement. 
Personalized federated learning (pFL) has emerged as a critical research direction to mitigate the heterogeneity across clients in FL. Per-FedAvg~\cite{fallah2020personalized} and pFedMe~\cite{t2020personalized} integrate Model-Agnostic Meta-Learning (MAML)~\cite{finn2017model} with FedAvg~\cite{mcmahan2017communication} to optimize a globally shared initialization, aiming to accelerate and enhance personalization on individual clients. Besides, pFedHN~\cite{shamsian2021personalized} learns a hyper-network to generate personalized models for each client. FedMD~\cite{li2019fedmd}, FedDF~\cite{lin2020ensemble}, and FedKT~\cite{li2020practical} adopt knowledge distillation frameworks to transfer globally shared public knowledge to local clients. Distinct from the global perspectives above, other works focus on local adaptation. Ditto~\cite{li2021ditto} introduces an additional module for each client to enable personalized adjustments on the global model. In FedRep~\cite{collins2021exploiting}, clients share a global feature extractor but learn their own local heads for personalized prediction. Additionally, several methods achieve personalization by combining the global model with the local one, such as FedALA~\cite{zhang2023fedala}, FedPHP~\cite{li2021fedphp}, and APFL~\cite{deng2020adaptive}.
% Federated recommendation (FR) has become a promising research direction due to its privacy-preserving advantages. FCF~\cite{ammad2019federated} is a pioneering FR work, where the collaborative filtering model is employed locally, and FedAvg method is adopted globally.
\subsection{Federated recommendation}
% In contrast to typical FL tasks that allow full model upload, federated recommendation (FR) inherently demands retaining local parameters for each client, \textit{i.e.}, the user embedding, necessitating specialized research. 

Pioneering works such as FedMF~\cite{chai2020secure} and FedNCF~\cite{perifanis2022federated} adapt the representative centralized models, \textit{i.e.}, Matrix Factorization (MF)~\cite{koren2009matrix} and Neural Collaborative Filtering (NCF)~\cite{he2017neural}, to the federated framework, respectively. Given the natural heterogeneity among clients, some efforts have been made in personalized federated recommendation (pFR). From a global perspective, GPFedRec~\cite{zhang2024gpfedrec} conducts the graph-guided aggregation to recover correlations among users, learning user-specific item embeddings. Besides, pFedGraph~\cite{ye2023personalized} performs global optimization based on similarity, deriving unique global models for each client. From a local view, PFedRec~\cite{zhang2023dual} designs a novel dual personalization mechanism to capture user preferences and refine global item embeddings. FedRAP~\cite{li2023federated} personalizes the global model locally by using an additional model. However, existing pFR methods focus on the personalization of global item embeddings, yet ignore the skew of local user embeddings caused by global aggregation, yielding suboptimal performance. To address this, we propose PFedCLR based on low-rank techniques to take both aspects into account. In contrast to LoRA-like methods~\cite{hu2022lora,nguyen2024towards} that are designed for efficient fine-tuning, PFedCLR introduces a dual-function mechanism tailored for FR scenario, which jointly achieves personalization and calibration, delivering a more effective and efficient solution.

% To address this, we propose PFedCLR to take both aspects into account. Some FR methods, such as FedCoLA~\cite{ding2024fedloca}, which is built upon FedRAP, utilize low-rank techniques primarily to reduce communication costs. In contrast, PFedCLR leverages low-rank techniques to support a dual-function mechanism that jointly achieves personalization and calibration, delivering a more effective and efficient solution.

\section{Problem definition}
\textbf{Notations.} Let $\mathcal{U}$ denote the set with $n$ users/clients, $\mathcal{I}$ the set with $m$ items. Each client $u\in\mathcal{U}$ keeps a local dataset  $\mathcal{D}_{u} = \{(u,i,r_{ui} | i\in\mathcal{I}_u)\}$, where $\mathcal{I}_u$ is the observed items of client $u$, $r_{ui}$ is the label. Particularly, we consider the typical recommendation task with implicit feedback, that is, $r_{ui}=1$ if client $u$ has interacted with item $i$, and $r_{ui}=0$ otherwise. Here, following mainstream methods~\cite{ammad2019federated,zhang2023dual,zhang2024beyond,li2025personalized,ding2024fedloca}, we adopt the embedding-based FR framework. Assuming the embedding dimension is $d$, client $u$ updates the user embedding $\mathbf{p}_u\in\mathbb{R}^d$ and item embeddings $\mathbf{Q}_u\in\mathbb{R}^{m \times d}$, while the server aggregates local models $\{\mathbf{Q}_u\}_{u\in\mathcal{U}}$ to obtain the global model $\mathbf{Q}_g$.

\textbf{FR objective.} The goal of FR is to predict $\hat{r}_{ui}$ of client $u$ for each unobserved item $i \in \mathcal{I} \setminus \mathcal{I}_u$ on local devices. Formally, the global optimization objective of FR tasks is
\begin{equation}\label{eq:frs_fl}
\min _{\left(\mathbf{p}_1, \mathbf{p}_2,\cdots, \mathbf{p}_n; \mathbf{Q}_1, \mathbf{Q}_2,\cdots, \mathbf{Q}_n\right)} \sum_{u=1}^n p_u \mathcal{L}_u\left(\mathbf{p}_u, \mathbf{Q}_u ; \mathcal{D}_u\right),
\end{equation}
\noindent where $p_u$ denotes the weight assigned to $\mathbf{Q}_u$ for global aggregation, \textit{e.g.}, $p_u=|\mathcal{D}_u|/ \sum_{v=1}^n|\mathcal{D}_v|$ in FedAvg \cite{mcmahan2017communication} and $p_u=1/n$ in FCF~\cite{ammad2019federated}. $\mathcal{L}_u$ is the local objective to facilitate the prediction of rating $\hat{r}_{ui}$. In this work, we employ the binary cross-entropy loss (BCE)~\cite{he2017neural} as $\mathcal{L}_u$ for the task with implicit feedback. Formally, the local optimization objective is defined as follows,
\begin{equation}\label{eq:bce}
    \mathcal{L}_{u}(\mathbf{p}_u, \mathbf{Q}_u ; \mathcal{D}_u)=-\sum_{(u, i, r_{u i}) \in \mathcal{D}_u} r_{u i} \log \hat{r}_{u i}+\left(1-r_{u i}\right) \log \left(1-\hat{r}_{u i}\right),
    \end{equation}
\noindent where $\hat{r}_{ui}=\sigma(\mathbf{p}_u^\top\mathbf{q}_i),\mathbf{q}_i\in\mathbf{Q}_u$ denotes the predicted rating of item $i$ by client $u$. Here, $\sigma(x)=1/(1+e^{-x})$ denotes the Sigmoid function to map the rating to $[0,1]$.

\textbf{User embedding skew.} We theoretically analyze the user embedding skew, and the detailed proof can be found in Appendix~\ref{app:user_skew}. According to Lemmas~\ref{thm:user_update_} and~\ref{thm:skew_quant_}, global aggregation introduces two additional terms into the original updating gradient of local user embedding: a scaling term and a shift term. The scaling term can alter the original gradient magnitude, while the shift term can change the original gradient direction. They jointly distort the optimization trajectory of local user embedding, leading to the skew observed in Figure~\ref{pic:pre_exp}. Furthermore, Lemma~\ref{thm:skew_bound_} quantifies the accumulated skew over $T$ global rounds, indicating that the user embedding finally converges to a suboptimal point influenced by global aggregation, with discussion provided in Remark~\ref{rmk:impact_skew}.
\begin{lemma}[User embedding update]
\label{thm:user_update_}
Without global aggregation, the original updating gradient of the local user embedding for client $u$ at round $t$ can be expressed as
\begin{equation}
    \nabla_u^{(t)}=\sum_{i=1}^m\left((\sigma(\mathbf{p}_u^{(t)\top} \mathbf{q}_i^{(u,t)}) - r_{ui}) \cdot \mathbf{q}_i^{(u,t)}\right)=\sum_{i=1}^mL_1^{(t)}\cdot \mathbf{q}_i^{(u,t)},
\end{equation}
\noindent where \scriptsize$\mathbf{q}_i^{(u, t)}$\normalsize~is the local $i$-th item embedding and \scriptsize$L_1^{(t)}$\normalsize$ = \sigma($\scriptsize$\mathbf{p}_u^{(t)\top} \mathbf{q}_i^{(u,t)}$\normalsize$) - r_{ui}$\normalsize~for notational simplicity.
\end{lemma}
\begin{lemma}[User embedding skew]
\label{thm:skew_quant_}
For global round $t$, the local user embedding skew of client $u$ introduced by global item embeddings is given by
\begin{equation}
\begin{split}
    \Delta_u^{(t)}\approx&\sum_{i=1}^m\left( \sigma^{\prime}(\mathbf{p}_u^{(t)\top} \mathbf{q}_i^{(u, t)}) \mathbf{p}_u^{(t)\top} \delta_i^{(t)} \cdot \mathbf{q}_i^{(u, t)}+(\sigma(\mathbf{p}_u^{(t)\top} \mathbf{q}_i^{(u,t)}) - r_{ui})\cdot\delta_i^{(t)}\right)\\
    =&\sum_{i=1}^m\left(\underbrace{L_2^{(t)} \mathbf{p}_u^{(t)\top} \delta_i^{(t)}\cdot\mathbf{q}_i^{(u,t)}}_{\text{Scaling Term}} + \underbrace{L_1^{(t)} \cdot\delta_i^{(t)}}_{\text{Shift Term}}\right),
\end{split}
\end{equation}
\noindent where \scriptsize$\delta_i^{(t)} = \mathbf{q}_i^{(g, t)} - \mathbf{q}_i^{(u, t)}$\normalsize~denotes the difference between the global $i$-th item embedding and the local one at round $t$. Additionally, we denote \scriptsize$L_2^{(t)}$\normalsize$ = \sigma^{\prime}($\scriptsize$\mathbf{p}_u^{(t)\top}\mathbf{q}_i^{(u,t)}$\normalsize$)$ for notational simplicity. Besides, $\sigma^{\prime}(x) = \sigma(x)(1 - \sigma(x))$ is the derivative of the Sigmoid function.
\end{lemma}

\begin{lemma}[Accumulated user embedding skew]
\label{thm:skew_bound_}
The final accumulated skew after the total global round $T$ is given by
\begin{equation}
    \|\Delta_{cumulative}^{(u,T)}\| \leq \sum_{i=1}^{m}\left(\frac{\eta C_1 \|\delta_i^{(0)}\|}{1 - \gamma} + \frac{\eta C_2 |\mathbf{p}_u^{(0)\top} \delta_i^{(0)}| \|\mathbf{q}_i^{(u, 0)}\|}{1 - \gamma}\right),
\end{equation}
% related to the predicted ratings
where $\gamma\in(0,1)$ is the difference amplification factor and $\eta$ is the local learning rate for embedding. Besides, $C_1$ and $C_2$ are two constants, defined as $C_1=\max\{ \sigma($\scriptsize$\mathbf{p}_u^{(t)\top} \mathbf{q}_i^{(u, t)}$\normalsize$) - r_{ui} \}_{t\in\{0,1,\cdots,T-1\}}$ and $C_2=\max\{\sigma^{\prime}($\scriptsize$\mathbf{p}_u^{(t)\top} \mathbf{q}_i^{(u, t)}$\normalsize$) \}_{t\in\{0,1,\cdots,T-1\}}$, respectively.
\end{lemma}

% \begin{figure*}[t]
% \centering
% \includegraphics[width=1.0\textwidth]{figs/trajectory_comparison.pdf} % Reduce the figure size so that it is slightly narrower than the column.
% \caption{Analysis of training trajectories for different FR methods.}
% \label{pic:trajectory_comparison}
% \end{figure*}

\section{Our proposed PFedCLR}
% \subsection{Overview}
The framework of PFedCLR is shown in Figure~\ref{pic:overview}. For each global round $t$, there are three key steps: 
\begin{itemize}[left=0pt, itemsep=-1pt]
\item\textbf{Step 1}: Each client $u$ downloads and updates the global model from last round, \scriptsize$\mathbf{Q}_g^{(t-1)}$\normalsize, to derive the local model \scriptsize$\mathbf{Q}_u^{(t)}$\normalsize. During this process, the user embedding is frozen at the state of last round, \scriptsize$\mathbf{p}_u^{(t-1)}$\normalsize, to prevent the potential skew introduced by \scriptsize$\mathbf{Q}_g^{(t-1)}$\normalsize. 

\item\textbf{Step 2}: The client freezes the local model \scriptsize$\mathbf{Q}_u^{(t)}$\normalsize, injecting and optimizing the zero-initialized low-rank matrices \scriptsize$\mathbf{A}_u^{(t)}\mathbf{B}_u^{(t)}$\normalsize~to buffer the impact of \scriptsize$\mathbf{Q}_u^{(t)}$\normalsize~, thereby enabling the calibration of local user embedding \scriptsize$\mathbf{p}_u^{(t)}$\normalsize. Furthermore, by merging the local model with the low-rank matrices, \textit{i.e.}, \scriptsize$\mathbf{Q}_u^{(t)}+\mathbf{A}_u^{(t)}\mathbf{B}_u^{(t)}$\normalsize, our method also achieves the personalization of item embeddings.

\item\textbf{Step 3}: The server aggregates the uploaded local models \scriptsize$\{\mathbf{Q}_u^{(t)}\}_{u\in\mathcal{U}}$\normalsize~to obtain the global model \scriptsize$\mathbf{Q}_g^{(t)}$\normalsize, distributing it to the clients for the next round $t+1$.
\end{itemize}

\begin{figure*}[t]
\centering
\includegraphics[width=0.99\textwidth]{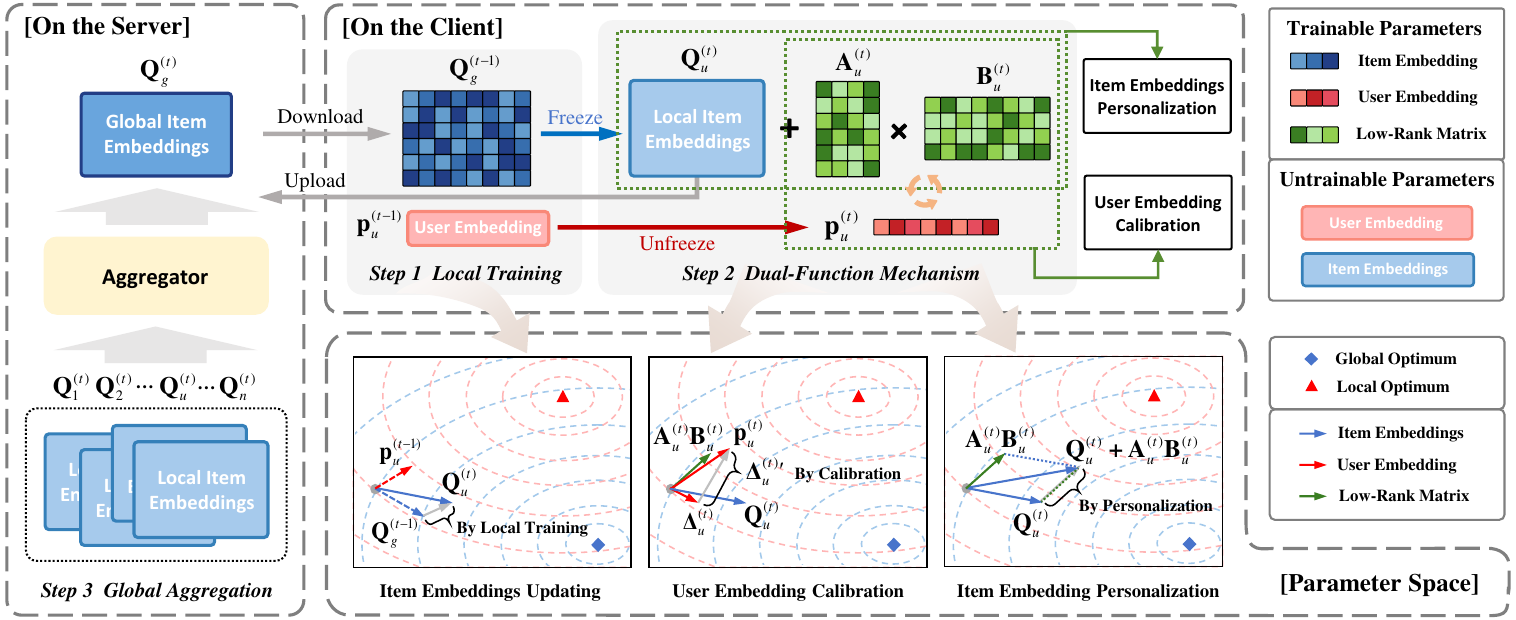} % Reduce the figure size so that it is slightly narrower than the column.
\caption{An overview of PFedCLR. During local training, the client only updates the global model \scriptsize$\mathbf{Q}_g^{(t-1)}$\normalsize~to obtain the local model \scriptsize$\mathbf{Q}_u^{(t)}$\normalsize. For dual-function mechanism, the low-rank matrices \scriptsize$\mathbf{A}_u^{(t)}\mathbf{B}_u^{(t)}$\normalsize~can provide \scriptsize$\Delta_u^{(t)\prime}$\normalsize~to mitigate the user embedding skew \scriptsize$\Delta_u^{(t)}$\normalsize~caused by global aggregation, thus calibrating the user embedding \scriptsize$\mathbf{p}_u^{(t)}$\normalsize. Besides, \scriptsize$\mathbf{A}_u^{(t)}\mathbf{B}_u^{(t)}$\normalsize~can also personalize the item embeddings \scriptsize$\mathbf{Q}_u^{(t)}$\normalsize.}
\label{pic:overview}
\end{figure*}

% \subsection{Our proposed PFedCLR} \scriptsize$\Delta_u^{(t-1)}$\normalsize~
\subsection{Step 1: Local training}
\textbf{Item embeddings updating.} At the beginning of each round $t$, client $u$ first downloads the aggregated global model \scriptsize$\mathbf{Q}_g^{(t-1)}$\normalsize~from the last round for local training. Based on the analysis in Lemma~\ref{thm:skew_quant_}, directly updating the local user embedding \scriptsize$\mathbf{p}_u^{(t-1)}$\normalsize~can result in the skew introduced by \scriptsize$\mathbf{Q}_g^{(t-1)}$\normalsize. Therefore, we freeze \scriptsize$\mathbf{p}_u^{(t-1)}$\normalsize~and only update \scriptsize$\mathbf{Q}_g^{(t-1)}$\normalsize~to obtain the local model \scriptsize$\mathbf{Q}_u^{(t)}$\normalsize~of this round, where the optimization objective is
\begin{equation}\label{eq:local_tra}
     \mathbf{Q}_u^{(t)}\leftarrow \min_{\mathbf{Q}_{u}}\mathcal{L}_{u}(\mathbf{p}_u^{(t-1)}, \mathbf{Q}_g^{(t-1)} ; \mathcal{D}_u).
\end{equation}

\subsection{Step 2: Dual-function mechanism}
\textbf{User embedding calibration.} We freeze the local model \scriptsize$\mathbf{Q}_u^{(t)}$\normalsize~since it is derived from \scriptsize$\mathbf{Q}_g^{(t-1)}$\normalsize, making its potential influence on user embedding controllable. Before introducing the efficiency-oriented design, we first present the formulation of injecting a full buffer matrix \scriptsize$\mathbf{W}_u^{(t)}$\normalsize~$\in\mathbb{R}^{m\times d}$ into \scriptsize$\mathbf{Q}_u^{(t)}$\normalsize~to calibrate the skew \scriptsize$\Delta_u^{(t)}$\normalsize. We provide a theoretical analysis of such calibration in Appendix~\ref{app:user_calibration}.

% we update the user embedding \scriptsize$\mathbf{p}_u^{(t-1)}$\normalsize~along the full buffer matrix \scriptsize$\mathbf{W}_u^{(t)}$\normalsize~$\in\mathbb{R}^{m\times d}$~injected into \scriptsize$\mathbf{Q}_u^{(t)}$\normalsize, calibrating the skew \scriptsize$\Delta_u^{(t)}$\normalsize. We provide a theoretical analysis of such calibration in Appendix~\ref{app:user_calibration}. 
\begin{lemma}[Calibration of user embedding skew]
\label{thm:skew_calibration_}
At global round $t$, the user embedding calibration for client $u$ achieved by our method can be approximated as
\begin{equation}
    \Delta_u^{(t){\prime}}\approx-\sum_{i=1}^{m}\left( \underbrace{\eta L_1^{(t)} L_2^{(t)} \mathbf{p}_u^{(t)\top} \mathbf{p}_u^{(t)} \cdot \mathbf{q}_i^{(u,t)}}_{\text{Scaling Term}} + \underbrace{\eta L_1^{(t)2} \cdot \mathbf{p}_u^{(t)}}_{\text{Shift Term}} \right).
\end{equation}
\end{lemma}
\begin{remark}[Dynamic regularization perspective for calibration]
By comparing the skew \scriptsize$\Delta_u^{(t)}$\normalsize~in Lemma~\ref{thm:skew_quant_} with our calibration \scriptsize$\Delta_u^{(t)\prime}$\normalsize~in Lemma~\ref{thm:skew_calibration_}, we explicitly introduce both a scaling term and a shift term, each tailored to mitigate the corresponding term of user embedding skew. Next, we focus on the shift term used to correct the gradient direction. Taking the $i$-th item for example, the shift term $-\eta$\scriptsize$ L_1^{(t)2} \mathbf{p}_u^{(t)}$\normalsize~is structurally similar to an $l_2$ regularization term applied to \scriptsize$\mathbf{p}_u^{(t)}$\normalsize, which is
\begin{equation}
    \frac{\partial}{\partial \mathbf{p}_u} \left( \frac{1}{2} \|\mathbf{p}_u^{(t)}\|^{2} \right) = \mathbf{p}_u^{(t)}.
\end{equation}
Hence, the shift term can be interpreted as an adaptive regularization term for the user embedding, with a coefficient of $-\eta$\scriptsize$ L_1^{(t)2}$\normalsize~that dynamically adjusts according to the user embedding skew at each round $t$. In this way, it serves as a self-regulating force that suppresses harmful gradient directions induced by global aggregation, offering robust user embedding updates.
\end{remark}

Formally, we update and obtain the calibrated user embedding \scriptsize$\mathbf{p}_u^{(t)}$\normalsize~and the buffer matrix \scriptsize$\mathbf{W}_u^{(t)}$\normalsize~by optimizing the following objective,
\begin{equation}\label{eq:p_W}
    \mathbf{p}_u^{(t)}, \mathbf{W}_{u}^{(t)}\leftarrow \min_{\mathbf{p}_{u},\mathbf{W}_u}\mathcal{L}_{u}(\mathbf{p}_u^{(t-1)}, (\mathbf{Q}_u^{(t)}+\mathbf{W}_u^{(t-1)}) ; \mathcal{D}_u).
\end{equation}

Considering limited local resources, using a full matrix as the buffer would double the local overhead. Inspired by LoRA~\cite{hu2022lora,nguyen2024towards}, we replace the zero-initialized matrix $\mathbf{W}_u$ with a low-rank decomposition $\mathbf{A}_u\mathbf{B}_u$, where $\mathbf{A}_u\in\mathbb{R}^{m\times r}$ is initialized as a zero matrix and $\mathbf{B}_u\in\mathbb{R}^{r\times d}$ is a randomly initialized Gaussian matrix, and $r \ll \min(m, d)$ denotes the rank of the decomposition. In this way, the client requires only $r(m + d) \ll m \times d$ additional parameters, towards improved efficiency and better suitability for practical deployment.
% II) Calibration effectiveness: $\mathbf{B}_u$ can provide a Gaussian basis of directional components, serving as a set of calibration directions, which enables effective calibration~\cite{hu2022lora, kopiczko2023vera}.
% \item Overfitting prevention: The low-rank constraint prevents overfitting to the local item deviation $\delta_i$, which could otherwise lead to a loss of shared global information. 
Thus, Equation~\ref{eq:p_W} can be rewritten as
\begin{equation}\label{eq:local_cal}
    \mathbf{p}_u^{(t)}, \mathbf{A}_{u}^{(t)}, \mathbf{B}_{u}^{(t)}\leftarrow \min_{\mathbf{p}_{u},\mathbf{A}_u,\mathbf{B}_u}\mathcal{L}_{u}(\mathbf{p}_u^{(t-1)}, (\mathbf{Q}_u^{(t)}+\mathbf{A}_u^{(t-1)}\mathbf{B}_u^{(t-1)}) ; \mathcal{D}_u).
\end{equation}
\textbf{Item embeddings personalization.} By merging the global information-related \scriptsize$\mathbf{Q}_u^{(t)}$\normalsize~with the local information-related \scriptsize$\mathbf{A}_u^{(t)} \mathbf{B}_u^{(t)}$\normalsize, we can derive the personalized item embeddings \scriptsize$\mathbf{Q}_u^{(t)}+\mathbf{A}_u^{(t)}\mathbf{B}_u^{(t)}$\normalsize~tailored for each client $u$, aiming at customized recommendation.

\subsection{Step 3: Global Aggregation}
On the server, we adopt the prevalent aggregator, FedAvg~\cite{mcmahan2017communication}, to perform global aggregation over the uploaded local models. This process can be formulated as
\begin{equation}\label{eq:global_agg}
\mathbf{Q}_g^{(t)}=\sum_{u=1}^{n}p_u\mathbf{Q}_u^{(t)}, \quad \text{where} \quad p_u=\frac{|\mathcal{D}_u|}{\sum_{v=1}^n |\mathcal{D}_v|}.
\end{equation}

\subsection{Discussion}
\textbf{Algorithm.} The pseudocode is provided in Appendix~\ref{app:algorithm}. Notably, Step 2 and Step 3 can be performed in parallel on the client and server, respectively, enhancing per-round training efficiency.

\textbf{Cost analysis.} With the employment of low-rank decomposition, PFedCLR incurs negligible overhead over the backbone FedMF. A detailed analysis can be found in Appendix~\ref{app:cost}.

\textbf{Privacy analysis.} I) Inherent privacy preservation: we upload \scriptsize$\mathbf{Q}_u^{(t)}$\normalsize~for aggregation after Step 1, while its personalization is performed in Step 2, which prevents the server from accessing user sensitive information. II) Enhanced privacy protection: we incorporate local differential privacy (LDP) into our method to further strengthen privacy. The privacy budget $\varepsilon$ is guaranteed by $\mathcal{S}_u / \lambda$, where $\lambda$ is the strength of Laplace noise and $\mathcal{S}_u$ is the global sensitivity of client $u$. We upper-bound $\mathcal{S}_u$ by $2p_u\eta C$, where $C$ is the clipping threshold. Detailed analysis is provided in Appendix~\ref{app:privacy}. 

% Extensive experiments in Appendix~\ref{app:privacy_exp} demonstrate the privacy-preserving effectiveness of our method.

\section{Experiment}\label{sec:experiment}

\subsection{Experimental setup}

\textbf{Datasets.} We verify our proposed PFedCLR on five recommendation benchmark datasets with varying scales and sparsity: \textbf{Filmtrust}~\cite{filmtrust_2013}, \textbf{Movielens-100K (ML-100K)}~\cite{harper2015movielens}, \textbf{Movielens-1M (ML-1M)}~\cite{harper2015movielens}, \textbf{HetRec2011}~\cite{cantador2011second} and \textbf{LastFM-2K}~\cite{cantador2011second}. More details are shown in appendix~\ref{app:datasets}.

\textbf{Evaluation protocols.} We follow the popular \textit{leave-one-out} evaluation~\cite{bayer2017generic}. The model performance is reported by \textit{Hit Ratio} (HR@$10$) and \textit{Normalized Discounted Cumulative Gain} (NDCG@$10$)~\cite{he2015trirank}. 
% In this work, we set $K=10$, and report the results as the average of 5 repeated experiments.

\textbf{Compared baselines.} We compare our method with both centralized and federated baselines. Specifically, MF~\cite{koren2009matrix}, NCF~\cite{he2017neural}, and LightGCN~\cite{he2020lightgcn} are effective centralized recommendation methods. In the federated setting, for a comprehensive comparison, we select traditional FR methods: FedMF~\cite{chai2020secure} and FedNCF~\cite{perifanis2022federated}; global aggregation-based pFR methods: pFedGraph~\cite{ye2023personalized} and GPFedRec~\cite{zhang2024gpfedrec}; and local adaptation-based pFR methods: PFedRec~\cite{zhang2023dual} and FedRAP~\cite{li2023federated}. More details about baselines can be found in appendix~\ref{app:baselines}.

\textbf{Implementation settings.} For a fair comparison, we set the global round $R=100$, batch size $B=256$, and embedding dimension $d=16$ for all methods. For our method, except for the extra parameters, \textit{i.e.}, the rank $r$ and learning rate $\beta$ of low-rank matrices, other hyper-parameters are kept consistent with those of the backbone FedMF. For other methods, we follow the experimental settings with the official code provided in the original paper. More details can be found in appendix~\ref{app:implementation}.

\begin{table}[ht]
\caption{Performance comparison on five datasets. "CenRS" and "FedRS" represent centralized and federated settings, respectively. The best FedRS results are bold and the second ones are underlined. Besides, "Improvement" indicates the performance improvement over the best baseline.}
\centering
\resizebox{\textwidth}{!}{
\begin{tabular}{llcccccccccc}
\toprule
& \multirow{2}{*}{\textbf{Method}}& \multicolumn{2}{c}{\textbf{Filmtrust}} & \multicolumn{2}{c}{\textbf{ML-100K}} & \multicolumn{2}{c}{\textbf{ML-1M}} & \multicolumn{2}{c}{\textbf{HetRec2011}} & \multicolumn{2}{c}{\textbf{LastFM-2K}} \\
\cmidrule(lr){3-4}\cmidrule(lr){5-6}\cmidrule(lr){7-8}\cmidrule(lr){9-10}\cmidrule(lr){11-12}
& & \scriptsize\textbf{HR@10} & \scriptsize\textbf{NDCG@10} & \scriptsize\textbf{HR@10} & \scriptsize\textbf{NDCG@10} & \scriptsize\textbf{HR@10} & \scriptsize\textbf{NDCG@10} & \scriptsize\textbf{HR@10} & \scriptsize\textbf{NDCG@10} & \scriptsize\textbf{HR@10} & \scriptsize\textbf{NDCG@10} \\
\midrule
\multirow{3}*{\textbf{CenRS}} & MF & 0.6886 & 0.5548 & 0.6543 & 0.3788 & 0.6088 & 0.3446 & 0.6275 & 0.3688 & 0.8440 & 0.6191 \\
& NCF & 0.6786 & 0.5437 & 0.6119 & 0.3422 & 0.5858 & 0.3267 & 0.6171 & 0.3663 & 0.7896 & 0.6069 \\
& LightGCN & 0.6956 & 0.5691 & 0.6787 & 0.3994 & 0.6684 & 0.3885 & 0.6611 & 0.3975 & 0.8448 & 0.6853 \\
\midrule
\multirow{6}*{\textbf{FedRS}} & FedMF & 0.6507 & 0.5171 & 0.4846 & 0.2723 & 0.4876 & 0.2734 & 0.5376 & 0.3206 & 0.5839 & 0.3930 \\
& FedNCF & 0.6497 & 0.5331 & 0.4252 & 0.2290 & 0.4180 & 0.2311 & 0.5083 & 0.2982 & 0.4933 & 0.3220 \\
\cmidrule(lr){2-12}
& pFedGraph & \underline{0.6956} & 0.4982 & 0.6204 & 0.4937 & 0.7262 & 0.5991 & 0.6962 & 0.5523 & 0.6485 & 0.6085 \\
& GPFedRec & 0.6866 & \underline{0.5578} & 0.6840 & 0.3982 & 0.6836 & 0.4012 & 0.6488 & 0.4016 & \textbf{0.7896} & 0.6499 \\
\cmidrule(lr){2-12}
& PFedRec & 0.6896 & 0.5479 & 0.6702 & 0.3929 & 0.6611 & 0.3849 & 0.6531 & 0.3948 & 0.7549 & \underline{0.6634} \\
& FedRAP & 0.6826 & 0.4628 & \underline{0.8823} & \underline{0.7980} & \underline{0.8661} & \underline{0.7666} & \underline{0.8486} &  \underline{0.6325} & 0.6257 & 0.5924 \\
\midrule
\textbf{Ours}& PFedCLR & \textbf{0.9102} & \textbf{0.7798} & \textbf{0.9989} & \textbf{0.9225} & \textbf{0.9603} & \textbf{0.8402} & \textbf{0.9522} & \textbf{0.8496} & \underline{0.7778} & \textbf{0.7164} \\
\midrule\midrule
\multicolumn{2}{c}{\textbf{Improvement}} & $\uparrow$ 27.02\% & $\uparrow$ 39.80\% & $\uparrow$ 13.22\% & $\uparrow$ 15.60\% & $\uparrow$ 10.88\% & $\uparrow$ 9.60\% & $\uparrow$ 12.21\% & $\uparrow$ 34.32\% & - & $\uparrow$ 7.99\% \\
\bottomrule
\end{tabular}
}
\label{tab:exp_main}
\end{table}

\subsection{Overall Comparisons}
\textbf{Performance.} The performance comparison of different methods is illustrated in Table~\ref{tab:exp_main}, based on which we provide the following analysis: \\
\textbf{I) PFedCLR significantly outperforms centralized methods across most datasets.} In the centralized setting, only user embeddings are typically treated as personalized parameters. In contrast, our method additionally personalizes item embeddings, enabling more customized recommendations for each client. Notably, LastFM-2K is an extremely sparse dataset with a vast number of items, which makes item embeddings insufficiently trained under the federated setting, yet PFedCLR only underperforms centralized methods at HR@10 while outperforming them at NDCG@10. \\
\textbf{II) PFedCLR achieves superior performance over SOTA federated methods across the majority of scenarios.} Based on traditional FR methods, pFR methods further personalize item embeddings for each client and thus achieve better performance. However, these methods all overlook the user embedding skew issue. In contrast, our method realizes both personalization and calibration with a dual-function mechanism, leading to the best performance. \\
\textbf{III) Our method demonstrates clear improvements on all evaluated datasets, highlighting its effectiveness.} Our method outperforms the strongest SOTA FR baselines on all five datasets, achieving HR@10 above 0.9 on the first four datasets, with the highest gain up to 39.8\%. Moreover, the notable improvement on the NDCG metric demonstrates a more accurate modeling of user preferences.

\begin{wrapfigure}{r}{0.45\textwidth}
\vspace{-0.4cm}
\centering
\includegraphics[width=0.44\textwidth]{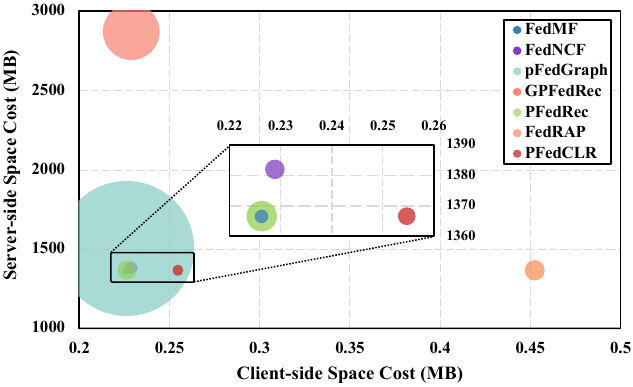} % Model performance and t-SNE training trajectory of local training with and without global aggregation.
\caption{Efficiency comparison on ML-1M. The position of each bubble indicates the space cost, while its size reflects the average training time per round. A smaller bubble denotes more efficient training per round.}
\vspace{-0.6cm}
\label{pic:efficiency}
\end{wrapfigure}

\textbf{Efficiency.}
We further compare the time and space costs of different FR methods. The results on ML-1M are shown in Figure~\ref{pic:efficiency}, and more experimental results can be found in Appendix~\ref{app:efficiency}. Compared to the backbone FedMF, PFedCLR introduces only a slight increase in space cost, adding approximately 0.03MB parameters locally. Moreover, unlike existing methods with alternating client-server updates, PFedCLR uploads the local model to the server immediately after Step 1. As a result, Step 2 and Step 3 can be executed in parallel on the client and the server, respectively, enabling improved round efficiency. Compared with other pFR methods, PFedCLR achieves a better trade-off between performance and efficiency.

\textbf{Convergence.} We compare the convergence of different FR methods, with detailed results presented in Appendix~\ref{app:convergence}. The main conclusions are: I) In the early training rounds, PFedCLR exhibits a notably fast convergence speed, outperforming most methods. II) In the later rounds, our method demonstrates a higher performance and maintains stable convergence. These results suggest both short-term efficiency and long-term effectiveness of PFedCLR compared to existing FR methods.

\subsection{Ablation Study}\label{sec:ablation}
\textbf{Component analysis.} To further evaluate the effectiveness of our method, we conduct a component-wise analysis based on the backbone FedMF, examining the following variants: I) Adaptation with Full matrix (AF): Following the pFR framework shown in Figure~\ref{pic:method_comparison} (b), we employ a full matrix $\mathbf{W}_u$ as the personalization mechanism to fine-tune the global item embeddings. II) Calibration with Full matrix (CF): This variant decouples the updates of user and item embeddings. A full buffer matrix $\mathbf{W}_u$ is employed within the dual-function mechanism to achieve both item embedding personalization and user embedding calibration. III) Calibration with Low-Rank matrices (CLR): This variant replaces the full matrix $\mathbf{W}_u$ in the dual-function mechanism with a low-rank decomposition $\mathbf{A}_u\mathbf{B}_u$. From the results reported in Table~\ref{tab:exp_ablation}, we have the following analysis:\\
\textbf{I)} \textbf{The personalization of global item embeddings can improve model performance.} AF fine-tunes item embeddings locally to provide customized recommendations for each client, which leads to noticeable performance gains. However, similar to existing pFR methods, it ignores the impact of global item embeddings on local user embeddings, resulting in suboptimal performance.\\
\textbf{II)} \textbf{The calibration of local user embedding further enhances performance.} CF uses a zero-initialized matrix to buffer the influence of global item embeddings on local user embeddings, providing more robust updates and mitigating the user embedding skew. This dual-function mechanism jointly personalizes item embeddings and calibrates user embedding, leading to superior results.\\
\textbf{III)} \textbf{The low-rank decomposition provides a more effective way to mitigate user embedding skew.} CLR decomposes the zero-initialized matrix $\mathbf{W}_u$ into a zero-initialized matrix $\mathbf{A}_u$ and a randomly initialized Gaussian matrix $\mathbf{B}_u$. Actually, $\mathbf{B}_u$ offers a low-rank subspace of potential calibration directions, allowing more suitable adjustments to the optimization trajectory of user embedding. 
% Moreover, this low-rank structure serves as a controlled fine-tuning mechanism, which helps prevent overfitting to the skew.

\begin{table}[t]
\caption{Ablation study results. "AF", "CF", and "CLR" denote Adaptation with Full matrix, Calibration with Full matrix, and Calibration with Low-Rank matrices, respectively.}
\centering
\resizebox{\textwidth}{!}{
\begin{tabular}{lcccccccccc}
\toprule
\multirow{2}{*}{\textbf{Method}}& \multicolumn{2}{c}{\textbf{Filmtrust}} & \multicolumn{2}{c}{\textbf{ML-100K}} & \multicolumn{2}{c}{\textbf{ML-1M}} & \multicolumn{2}{c}{\textbf{HetRec2011}} & \multicolumn{2}{c}{\textbf{LastFM-2K}} \\
\cmidrule(lr){2-3}\cmidrule(lr){4-5}\cmidrule(lr){6-7}\cmidrule(lr){8-9}\cmidrule(lr){10-11}
&\scriptsize\textbf{HR@10} & \scriptsize\textbf{NDCG@10} & \scriptsize\textbf{HR@10} & \scriptsize\textbf{NDCG@10} & \scriptsize\textbf{HR@10} & \scriptsize\textbf{NDCG@10} & \scriptsize\textbf{HR@10} & \scriptsize\textbf{NDCG@10} & \scriptsize\textbf{HR@10} & \scriptsize\textbf{NDCG@10} \\
\midrule
FedMF & 0.6507 & 0.5171 & 0.4846 & 0.2723 & 0.4912 & 0.2751 & 0.5376 & 0.3206 & 0.5839 & 0.3930 \\
FedMF w/ AF & 0.6717 & 0.5429 & 0.6151 & 0.4307 & 0.5897 & 0.4007 & 0.6947 & 0.5024 & 0.6612 & 0.4600 \\
FedMF w/ CF & 0.7754 & 0.6313 & 0.8452 & 0.6955 & 0.8349 & 0.6794 & 0.8239 & 0.6612 & 0.7667 & 0.7037 \\
FedMF w/ CLR & \textbf{0.9102} & \textbf{0.7798} & \textbf{0.9989} & \textbf{0.9225} & \textbf{0.9603} & \textbf{0.8402} & \textbf{0.9522} & \textbf{0.8496} & \textbf{0.7778} & \textbf{0.7164} \\
\bottomrule
\end{tabular}
}
\label{tab:exp_ablation}
\end{table}

\begin{wrapfigure}{r}{0.37\textwidth}
\vspace{-0.4cm}
\centering
\includegraphics[width=0.37\textwidth]{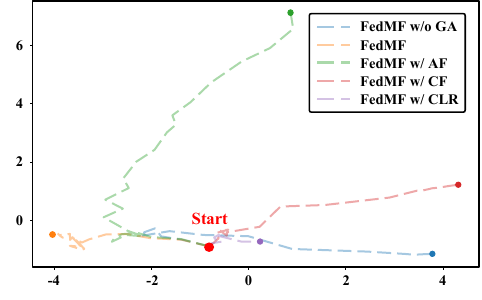} % Model performance and t-SNE training trajectory of local training with and without global aggregation.
\vspace{-0.4cm}
\caption{Training trajectories of local user embedding for Client \#94. "GA" denotes global aggregation.}
\vspace{-0.4cm}
\label{pic:trajectory}
\end{wrapfigure}

\textbf{Trajectory visualization.} The t-SNE training trajectories of local user embedding under different variants are illustrated in Figure~\ref{pic:trajectory}. More cases and analyses can be found in Appendix~\ref{app:trajectory}. We can observe that: I) FedMF exhibits a noticeable difference with and without global aggregation, highlighting the user embedding skew issue. II) AF explores client-specific user embedding updates via local personalization, but these updates are still influenced by global aggregation. III) CF effectively calibrates the update direction of user embedding. However, it greatly weakens the influence of global aggregation, resulting in the loss of global collaborative information, which hinders optimal performance. IV) Due to the inherent low-rank constraints, CLR strikes a better balance between calibration and preserving global information, yielding superior results.

\textbf{Embedding visualization.} We further visualize the user and item embeddings of FedMF and PFedCLR with t-SNE. As one case is shown in Figure~\ref{pic:embedding}, we have the following insights: I) Calibration:
For FedMF, the user embedding is clearly influenced by item embeddings $\mathbf{Q}_u$, confirming that global aggregation can cause the user embedding skew issue. In contrast, PFedCLR effectively mitigates this issue, as illustrated in subfigures (a) and (b). II) Personalization:
Beyond calibrating the user embedding, the low-rank matrices $\mathbf{A}_u\mathbf{B}_u$ of PFedCLR can also personalize $\mathbf{Q}_u$, with the user exhibiting distinct preferences toward the interacted items, as shown in subfigures (c) and (d). More cases and detailed analysis are provided in Appendix~\ref{app:embedding}.

\begin{figure*}[htbp]
\centering
\includegraphics[width=1.0\textwidth]{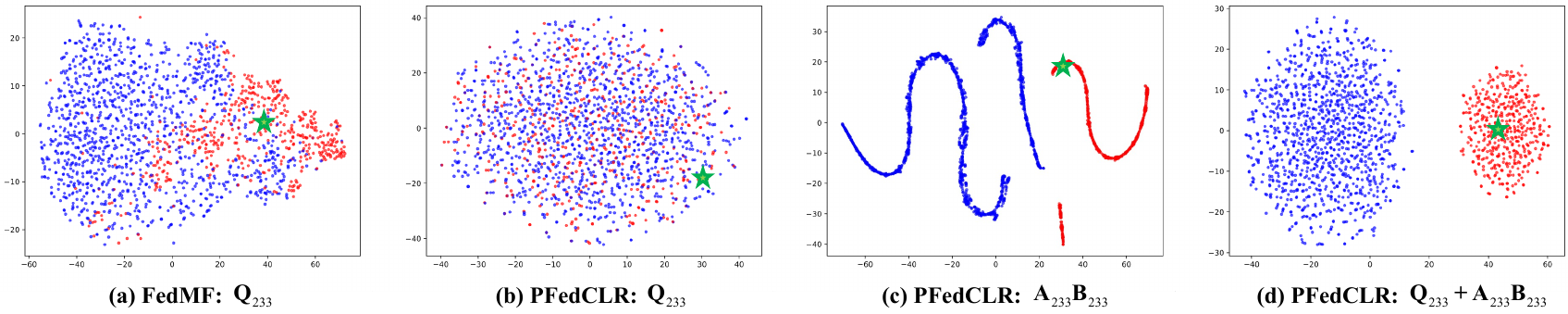} % Reduce the figure size so that it is slightly narrower than the column.
\caption{Visualizations of embeddings for Client \#233. The \textcolor{green}{user} is marked by the \textcolor{green}{green} pentagram. Items \textcolor{red}{interacted} with and \textcolor{blue}{not interacted} with by the user are indicated in \textcolor{red}{red} and \textcolor{blue}{blue}, respectively.}
\label{pic:embedding}
\end{figure*}

\subsection{Hyper-parameter Analysis}

\begin{figure*}[htbp]
\centering
\includegraphics[width=0.99\textwidth]{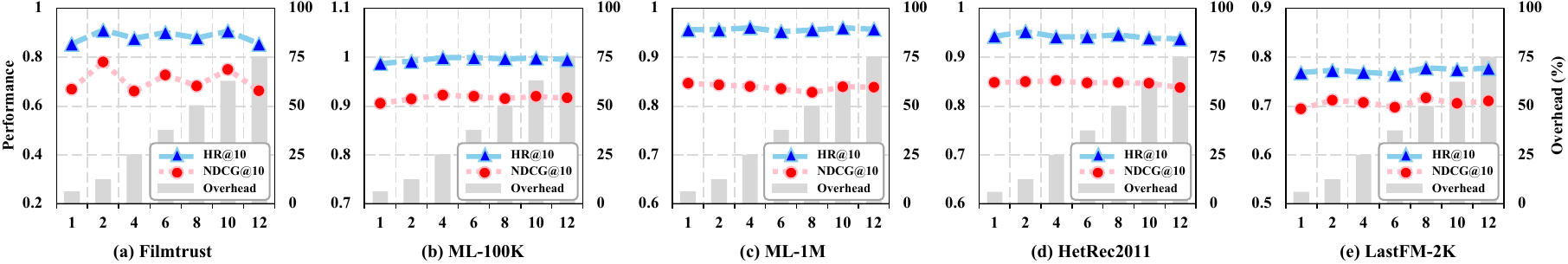} % Reduce the figure size so that it is slightly narrower than the column.
\caption{Model performance and extra overhead under different rank. The horizontal axis is the rank of low-rank matrices. The left vertical axis indicates the model performance, \textit{i.e.}, HR@10 and NDCG@10, while the right one indicates the incremental space cost relative to the backbone FedMF.}
\label{pic:hyper_rank}
\end{figure*}

Our method is simple to tune, requiring only two additional hyper-parameters: I) \textbf{Rank $r$ of the low-rank matrices.} As shown in Figure~\ref{pic:hyper_rank}, when the rank $r$ is small, \textit{e.g.}, $r=2$, our method achieves outstanding performance across all datasets with negligible additional overhead, demonstrating its promising practical potential. II) \textbf{Learning rate $\beta$ of the low-rank matrices.} A larger or smaller $\beta$ is detrimental to both personalization and calibration, while the best performance is achieved when $\beta$ equals the default learning rate $\eta$ for embedding, \textit{i.e.}, $\beta = \eta = 0.01$. Detailed results and analysis can be found in Appendix~\ref{app:parameter}.

% [htbp]
\subsection{Privacy Protection}
\begin{figure}[htbp]
\centering
\includegraphics[width=0.99\textwidth]{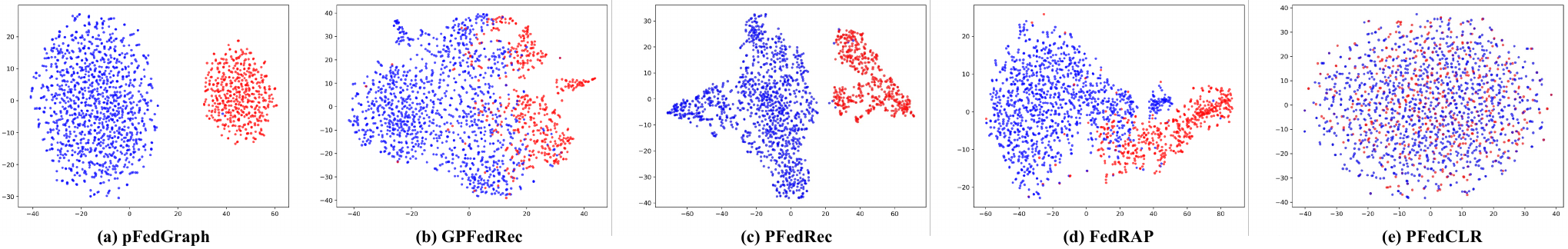}
\caption{Visualizations of item embeddings uploaded to the server by different pFR methods. Items \textcolor{red}{interacted} with and \textcolor{blue}{not interacted} with by the user are indicated in \textcolor{red}{red} and \textcolor{blue}{blue}, respectively.}
\label{pic:privacy}
\end{figure}

\begin{wraptable}{r}{0.42\textwidth}
\centering
\vspace{-0.65cm}
\caption{Results of applying local differential privacy (LDP) into our method.}
\vspace{0.2cm}
\resizebox{0.42\textwidth}{!}{
\begin{tabular}{llcccc}
\toprule
\textbf{Dataset} & \textbf{Metrics} & \textbf{w/o LDP} & \textbf{w/ LDP} & \textbf{Degradation} \\
\midrule
\multirow{2}{*}{\textbf{Filmtrust}} & HR@10 & 0.9102 & 0.9122 & - \\
 & NDCG@10 & 0.7798 & 0.7741 & $\downarrow$ 0.73\%  \\
\midrule
\multirow{2}{*}{\textbf{ML-100K}} & HR@10 & 0.9989 & 0.9979 & $\downarrow$ 0.10\% \\
 & NDCG@10 & 0.9225 & 0.9215 & $\downarrow$ 0.11\% \\
\midrule
\multirow{2}{*}{\textbf{ML-1M}} & HR@10 & 0.9603 & 0.9586 & $\downarrow$ 0.18\% \\
 & NDCG@10 & 0.8402 & 0.8379 & $\downarrow$ 0.27\% \\
\midrule
\multirow{2}{*}{\textbf{HetRec2011}} & HR@10 & 0.9522 & 0.9460  & $\downarrow$ 0.65\% \\
 & NDCG@10 & 0.8496 & 0.8405 & $\downarrow$ 1.07\% \\
\midrule
\multirow{2}{*}{\textbf{LastFm-2K}} & HR@10 & 0.7778 & 0.7526 & $\downarrow$ 3.24\% \\
 & NDCG@10 & 0.7164 & 0.6816 & $\downarrow$ 4.86\% \\
\bottomrule
\end{tabular}
}
\label{tab:exp_ldp}
% \vspace{-0.4cm}
\end{wraptable}

I) Inherent privacy preservation: The visualization results of the uploaded item embeddings are shown in Figure~\ref{pic:privacy}. We can observe that existing pFR methods perform personalization before uploading the local model, which leads to the leakage of user preference information, as demonstrated by the clear distinction between interacted and non-interacted items. In contrast, our method performs personalization after uploading the model, thereby protecting user-sensitive information.\\
II) Enhanced privacy protection: We further incorporate LDP to strengthen privacy, as shown in Table~\ref{tab:exp_ldp}. It is evident that our method remains robust when applying LDP and still outperforms other SOTA methods. Detailed privacy analysis can be found in Appendix~\ref{app:privacy_exp}.

\section{Conclusion}
In this paper, we first empirically identify the local user embedding skew caused by globally aggregated item embeddings in federated recommendation. Next, we theoretically analyze the rationale behind such skew. To address this issue, we propose PFedCLR, which integrates a dual-function mechanism to simultaneously calibrate local user embedding and personalize global item embeddings, significantly enhancing model performance. Furthermore, considering client-side resource constraints, we apply low-rank decomposition to the buffer matrix within this mechanism, achieving improved efficiency. To preserve privacy, personalization is conducted after uploading the local model, shielding user-sensitive information from the server. Extensive experiments demonstrate that PFedCLR effectively mitigates user embedding skew and achieves a well-balanced trade-off among performance, efficiency, and privacy compared to other SOTA methods.

% \section*{References}
\bibliographystyle{unsrt}
\bibliography{ref}

%%%%%%%%%%%%%%%%%%%%%%%%%%%%%%%%%%%%%%%%%%%%%%%%%%%%%%%%%%%%
\newpage
\appendix
\section*{Appendix}
To support the main content of the paper, we provide the following supplementary materials:
\begin{itemize}[left=0pt, itemsep=-1pt]
\item \textbf{Algorithms}: pseudocode of PFedCLR.
\item \textbf{Theoretical Analysis}: user embedding skew, user embedding calibration, cost analysis, and privacy analysis.
\item \textbf{Experimental Details}: datasets, baselines, and implementation.
\item \textbf{Extensive Experiment Results}: motivation, efficiency, convergence, effectiveness (trajectory and embedding visualization), hyper-parameter analysis, and privacy protection.
\item \textbf{Limitations}: limitations of PFedCLR.
\end{itemize}

\section{Algorithms}\label{app:algorithm}
% \IncMargin{1em}
% % \SetNlSty{textbf}{}{\hspace{-1.5em}}
% \begin{algorithm}[H]
% \caption{PFedCLR}
% \label{alg:PFedCLR}
% \SetAlgoNoLine
% \SetKwInOut{Input}{Input}\SetKwInOut{Output}{Output}
% \Input{ clients $\mathcal{U}$; global rounds \(T\); local epochs \(E\); batch size \(B\); learning rates \(\eta,\beta\); rank \(r\)}
% \Output{ local models for each client $u$: \(\mathbf{p}_u^{(T)},\mathbf{Q}_u^{(T)},\mathbf{A}_u^{(T)}\mathbf{B}_u^{(T)}\)}
% \BlankLine
% \textbf{Global Procedure:}
% \SetAlgoLined

% Initialize global model $\mathbf{Q}_g^{(0)}$\;
% \For{each client $u=1,2,\cdots,n$ \textbf{in parallel}}{
% Initialize local models $\mathbf{p}_u^{(0)},\mathbf{Q}_u^{(0)},\mathbf{A}_u^{(0)}\mathbf{B}_u^{(0)}$\;
% }
% % \vspace{-0.4cm}
% \For{each global round $t=1,2,\cdots, T$}{
% $\mathcal{U}_s \leftarrow$ Randomly select $n_s$ clients from $\mathcal{U}$\;
% \For{each client $u\in\mathcal{U}_s$ \textbf{in parallel}}{
% Downloads $\mathbf{Q}_g^{(t-1)}$ from the server\;
% $\mathbf{Q}_u^{(t)}\leftarrow$ ClientUpdate($\mathbf{Q}_g^{(t-1)}$,\textit{u})\;
% }
% $\mathbf{Q}_g^{(t)}\leftarrow$ Perform Equation~(\ref{eq:global_agg}); \hspace{0.5em} $\triangleright$ Global Aggregation
% }
% \end{algorithm}
\begin{algorithm}[htbp]
\caption{PFedCLR}
\label{alg:PFedCLR}
\raggedright
\textbf{Input:} clients $\mathcal{U}$; global rounds \(T\); local epochs \(E\); batch size \(B\); learning rates \(\eta,\beta\); rank \(r\) \\
\textbf{Output:} \(\mathbf{p}_u^{(T)},\mathbf{Q}_u^{(T)},\mathbf{A}_u^{(T)}\mathbf{B}_u^{(T)}\) for each client $u$ \\
\vspace{1ex}
\textbf{Global Procedure:}\\
\begin{algorithmic}[1]
\STATE Initialize global model $\mathbf{Q}_g^{(0)}$; \\
\FOR{each client $u=1,2,\cdots,n$ \textbf{in parallel}}
    \STATE Initialize local models $\mathbf{p}_u^{(0)},\mathbf{Q}_u^{(0)},\mathbf{A}_u^{(0)}\mathbf{B}_u^{(0)}$; \hspace{3em} $\triangleright$ Initialization
\ENDFOR
\FOR{each global round $t=1,2,\cdots, T$}
    \STATE $\mathcal{U}_s \leftarrow$ Randomly select $n_s$ clients from $\mathcal{U}$;
    \FOR{each client $u\in\mathcal{U}_s$ \textbf{in parallel}}
        \STATE Downloads $\mathbf{Q}_g^{(t-1)}$ from the server; \\
        \STATE $\mathbf{Q}_u^{(t)}\leftarrow$ ClientUpdate($\mathbf{Q}_g^{(t-1)}$,\text{u}); \\
    \ENDFOR
    \STATE $\mathbf{Q}_g^{(t)}\leftarrow$ Aggregate $\{\mathbf{Q}_u^{(t)}\}_{u\in\mathcal{U}_s}$ by Equation~(\ref{eq:global_agg}); \hspace{1em} $\triangleright$ Step 3: Global Aggregation
\ENDFOR
\end{algorithmic}
\textbf{ClientUpdate}$(\mathbf{Q}_g^{(t-1)}, $u$)$:\\
\begin{algorithmic}[1]
\FOR{each local epoch $e=1,2,\cdots, E$}
    \FOR{each batch $b=1,2,\cdots,B$ in $\mathcal{D}_u$}
        \STATE $\mathbf{Q}_u^{(t)}\leftarrow$ Update $\mathbf{Q}_g^{(t-1)}$ with Equation~(\ref{eq:local_tra}); \hspace{3em} $\triangleright$ Step 1: Local Training
    \ENDFOR
\ENDFOR
\STATE Uploads $\mathbf{Q}_u^{(t)}$ to the server; \\
\FOR{each local epoch $e=1,2,\cdots, E$}
    \FOR{each batch $b=1,2,\cdots,B$ in $\mathcal{D}_u$}
        \STATE $\mathbf{p}_u^{(t)},\mathbf{A}_u^{(t)}\mathbf{B}_u^{(t)}\leftarrow$ Optimize $\mathbf{p}_u^{(t-1)},\mathbf{A}_u^{(t-1)}\mathbf{B}_u^{(t-1)}$ with Equation~(\ref{eq:local_cal});\\ \hspace{20em} $\triangleright$ Step 2: Dual-Function Mechanism
    \ENDFOR
\ENDFOR
\end{algorithmic}
\end{algorithm}

The overall procedure of the proposed PFedCLR is summarized in Algorithm~\ref{alg:PFedCLR}. Prior to the start of global training, each client initializes its local model based on the global parameter \scriptsize$\mathbf{Q}_g^{(0)}$\normalsize. For each global round $t$, clients first perform Step 1 to update item embeddings \scriptsize$\mathbf{Q}_u^{(t)}$\normalsize~locally, followed by Step 2, where the dual-function mechanism calibrates the user embedding \scriptsize$\mathbf{p}_u^{(t)}$\normalsize~and updates the personalized low-rank matrices \scriptsize$\mathbf{A}_u^{(t)}\mathbf{B}_u^{(t)}$\normalsize. The server then performs Step 3 to aggregate uploaded models \scriptsize$\{\mathbf{Q}_u^{(t)}\}_{u\in\mathcal{U}_s}$\normalsize~and update the global parameter \scriptsize$\mathbf{Q}_g^{(t)}$\normalsize.

Notably, clients can upload their models immediately after Step 1, allowing Step 2 on the client side and Step 3 on the server side to proceed in parallel. This parallelism improves training efficiency compared to conventional methods that rely on alternating client-server updates. Furthermore, since personalization occurs locally in Step 2, only non-personalized local models are uploaded after Step 1, thereby preserving user sensitive information.

\section{Theoretical Analysis}\label{app:theoretical}

\subsection{User embedding skew}\label{app:user_skew}

\begin{lemma}[User embedding update]
% \label{thm:user_update_}
Without global aggregation, the original updating gradient of the local user embedding for client $u$ at round $t$ can be expressed as
\begin{equation}
    \nabla_u^{(t)}=\sum_{i=1}^m\left((\sigma(\mathbf{p}_u^{(t)\top} \mathbf{q}_i^{(u,t)}) - r_{ui}) \cdot \mathbf{q}_i^{(u,t)}\right)=\sum_{i=1}^mL_1^{(t)}\cdot \mathbf{q}_i^{(u,t)},
\end{equation}
\noindent where $\mathbf{q}_i^{(u, t)}$ is the local $i$-th item embedding of client $u$ and we denote $L_1^{(t)} = \sigma(\mathbf{p}_u^{(t)\top} \mathbf{q}_i^{(u,t)}) - r_{ui}$ for notational simplicity.
\end{lemma}
\begin{proof}For simplicity, we omit the round notation $t$ in this proof and illustrate with the $i$-th item embedding $\mathbf{q}_{i}\in\mathbb{R}^d$ for example. Firstly, for local training, the updating gradient of user embedding before global aggregation is
\begin{equation}\label{eq:user_update}
\begin{split}
    \frac{\partial \mathcal{L}_u}{\partial \mathbf{p}_u} &= \frac{\partial \mathcal{L}_u}{\partial \hat{r}_{ui}}\frac{\partial \hat{r}_{ui}}{\partial \mathbf{p}_u} \\
    &= \left(-\frac{r_{ui}}{\hat{r}_{ui}} + \frac{1 - r_{ui}}{1 - \hat{r}_{ui}}\right) \cdot \sigma^{\prime}(\mathbf{p}_u^\top \mathbf{q}_i^{(u)}) \mathbf{q}_i^{(u)} \\
    &=\left(-\frac{r_{ui}}{\sigma(\mathbf{p}_u^\top \mathbf{q}_i^{(u)})} + \frac{1 - r_{ui}}{1 - \sigma(\mathbf{p}_u^\top \mathbf{q}_i^{(u)})}\right) \cdot \sigma(\mathbf{p}_u^\top \mathbf{q}_i^{(u)})(1-\sigma(\mathbf{p}_u^\top \mathbf{q}_i^{(u)})) \mathbf{q}_i^{(u)} \\
    &=(\sigma(\mathbf{p}_u^\top \mathbf{q}_i^{(u)}) - r_{ui}) \mathbf{q}_i^{(u)}.
\end{split}
\end{equation}
% \noindent where we denote $\mathbf{p}_u$ and $\mathbf{q}_i$ as $\mathbf{p}_u^{(t)}$ and $\mathbf{q}_i^{(u, t)}$ for simplicity, respectively.
% Similarly, the gradient with respect to the item embedding $\mathbf{q}_i$ is given by
% \begin{equation}
%     \frac{\partial \mathcal{L}_u}{\partial \mathbf{q}_i} = (\sigma(\mathbf{p}_u^\top \mathbf{q}_i) - r_{ui}) \mathbf{p}_u.
% \end{equation}
Hence, the proof is complete.\end{proof}

\begin{lemma}[User embedding skew]
\label{thm:skew_quant}
For global round $t$, the local user embedding skew of client $u$ introduced by global item embeddings is given by
\begin{equation}
    \Delta_u^{(t)}\approx\sum_{i=1}^m\left( \sigma^{\prime}(\mathbf{p}_u^{(t)\top} \mathbf{q}_i^{(u, t)}) \mathbf{p}_u^{(t)\top} \delta_i^{(t)} \cdot \mathbf{q}_i^{(u, t)}+(\sigma(\mathbf{p}_u^{(t)\top} \mathbf{q}_i^{(u,t)}) - r_{ui})\cdot\delta_i^{(t)}\right),
\end{equation}
\noindent where $\mathbf{q}_i^{(u,t)}$ denotes the embedding of item $i$ for client $u$ locally at round  $t$, $\sigma^{\prime}(x) = \sigma(x)(1 - \sigma(x))$ is the derivative of the Sigmoid function. Beside, $\delta_i^{(t)} = \mathbf{q}_i^{(g, t)} - \mathbf{q}_i^{(u, t)}$ denotes the difference between the global $i$-th item embedding and the local one at round $t$. Let $L_1^{(t)} = \sigma(\mathbf{p}_u^{(t)\top} \mathbf{q}_i^{(u,t)}) - r_{ui}$, $L_2^{(t)} = \sigma^{\prime}(\mathbf{p}_u^{(t)\top}\mathbf{q}_i^{(u,t)})$ for notational simplicity, this theorem can be formulated as
\begin{equation}
    \Delta_u^{(t)}\approx\sum_{i=1}^m\left(\underbrace{L_2^{(t)} \mathbf{p}_u^{(t)\top} \delta_i^{(t)}\cdot\mathbf{q}_i^{(u,t)}}_{\text{Scaling Term}} + \underbrace{L_1^{(t)} \cdot\delta_i^{(t)}}_{\text{Shift Term}}\right).
\end{equation}
\end{lemma}

\begin{proof}After the client uploads the local item embeddings, the server performs global aggregation, introducing the difference between global item embedding $\mathbf{q}_i^{(g)}$ and local item embedding $\mathbf{q}_i^{(u)}$, which is defined as
\begin{equation}
    \delta_i = \mathbf{q}_i^{(g)} - \mathbf{q}_i^{(u)}.
\end{equation}
Next, the client downloads the global model $\mathbf{q}_i^{(g)}$ and performs local training, we have the updating gradient of user embedding after global aggregation as
\begin{equation}
\begin{split}
    \frac{\partial \mathcal{L}_u}{\partial \mathbf{p}_u} = &(\sigma(\mathbf{p}_u^\top \mathbf{q}_i^{(g)}) - r_{ui}) \mathbf{q}_i^{(g)} \\
    = &(\sigma(\mathbf{p}_u^\top (\mathbf{q}_i^{(u)}+\delta_i)) - r_{ui}) (\mathbf{q}_i^{(u)}+\delta_i) \\
    \overset{\text{Taylor}}{\underset{\text{Expansion}}{=}} & (\sigma(\mathbf{p}_u^\top \mathbf{q}_i^{(u)}) + \sigma^{\prime}(\mathbf{p}_u^\top \mathbf{q}_i^{(u)}) \mathbf{p}_u^\top \delta_i+o(||\delta_i||^2)-r_{ui})(\mathbf{q}_i^{(u)}+\delta_i) \\
    = & \underbrace{(\sigma(\mathbf{p}_u^\top \mathbf{q}_i^{(u)}) - r_{ui}) \mathbf{q}_i^{(u)}}_{\text{Original Gradient Term}}+\underbrace{\sigma^{\prime}(\mathbf{p}_u^\top\mathbf{q}_i^{(u)})\mathbf{p}_u^\top \delta_i\mathbf{q}_i^{(u)}}_{\text{Scaling Term}}+\underbrace{(\sigma(\mathbf{p}_u^\top \mathbf{q}_i^{(u)}) - r_{ui})\delta_i}_{\text{Shift Term}}+\mathcal{O}(||\delta_i||^2),
\end{split}
\end{equation}
where the first original gradient term represents the updating gradient of user embedding without global aggregation as shown in Equation~(\ref{eq:user_update}). The second scaling term indicates that $\delta_i$ can change the original gradient magnitude. The third shift term shows that $\delta_i$ can change the original gradient direction. The final term is the higher-order infinitesimal term. Due to the difference $\delta_i$ introduced by global aggregation, the scaling and shift terms jointly distort the training trajectory of the local user embedding, resulting in the suboptimal results shown in Figure~\ref{pic:pre_exp}. Therefore, ignoring the higher-order infinitesimal term, the user embedding skew is given by
\begin{equation}
    \Delta_u\approx\sum_{i=1}^{m}\left(\sigma^{\prime}(\mathbf{p}_u^\top\mathbf{q}_i^{(u)})\mathbf{p}_u^\top \delta_i\cdot\mathbf{q}_i^{(u)}+(\sigma(\mathbf{p}_u^\top \mathbf{q}_i^{(u)}) - r_{ui})\cdot\delta_i\right).
\end{equation}
Hence, the proof is complete.\end{proof}

\begin{lemma}[Accumulated user embedding skew]
\label{thm:skew_bound}
Suppose the norm of the difference $\delta_i^{(t)}$ can grow by at most a factor of $\gamma$ in each round, \textit{i.e.}, $\|\delta_i^{(t+1)}\| \leq \gamma \|\delta_i^{(t)}\|$, where $\gamma \in (0,1) $ is the difference amplification factor. The final accumulated skew after global round $T$ is given by
\begin{equation}
    \|\Delta_{cumulative}^{(u,T)}\| \leq \sum_{i=1}^{m}\left(\frac{\eta C_1 \|\delta_i^{(0)}\|}{1 - \gamma} + \frac{\eta C_2 |\mathbf{p}_u^{(0)\top} \delta_i^{(0)}| \|\mathbf{q}_i^{(u, 0)}\|}{1 - \gamma}\right),
\end{equation}
% related to the predicted ratings
where $C_1$ and $C_2$ are two constants, defined as $C_1=\max\{ \sigma(\mathbf{p}_u^{(t)\top} \mathbf{q}_i^{(u, t)}) - r_{ui} \}_{t\in\{0,1,\cdots,T-1\}}$ and $C_2=\max\{\sigma^{\prime}(\mathbf{p}_u^{(t)\top} \mathbf{q}_i^{(u, t)}) \}_{t\in\{0,1,\cdots,T-1\}}$, respectively. Besides, $\eta$ denotes the local learning rate for embedding.
\end{lemma}

% Here, we also take the $i$-th item embedding for example.
\begin{proof}Following Lemma~\ref{thm:skew_quant}, the accumulated local user embedding skew at global round $t+1$ for client $u$ is
\begin{equation}\label{eq:error_gain}
    \Delta_{cumulative}^{(u,t+1)} = \Delta_{cumulative}^{(u,t)} - \eta\sum_{i=1}^{m} \left[ (\sigma(\mathbf{p}_u^{(t)\top} \mathbf{q}_i^{(u, t)}) - r_{ui}) \delta_i^{(t)} + \sigma^{\prime}(\mathbf{p}_u^{(t)\top} \mathbf{q}_i^{(u, t)}) (\mathbf{p}_u^{(t)\top} \delta_i^{(t)}) \mathbf{q}_i^{(u, t)} \right],
\end{equation}
where $\eta$ denotes the local learning rate for embedding. Next, we consider an upper bound on the increased skew at round $t+1$ as follows,
\begin{equation}
    \|\Delta_{cumulative}^{(u,t+1)}\| \leq \|\Delta_{cumulative}^{(u,t)}\| + \eta \sum_{i=1}^{m} \left( C_1 \|\delta_i^{(t)}\| + C_2 |\mathbf{p}_u^{(t)\top} \delta_i^{(t)}| \|\mathbf{q}_i^{(u, t)}\| \right),
\end{equation}
where $C_1$ and $C_2$ are two constants related to the total global round $T$, defined from Equation~(\ref{eq:error_gain}) as $C_1=\max\{ \sigma(\mathbf{p}_u^{(t)\top} \mathbf{q}_i^{(u, t)}) - r_{ui} \}_{t\in\{0,1,\cdots,T-1\}}$ and $C_2=\max\{\sigma^{\prime}(\mathbf{p}_u^{(t)\top} \mathbf{q}_i^{(u, t)}) \}_{t\in\{0,1,\cdots,T-1\}}$, respectively. Assuming the norm of the difference $\delta_i^{(t)}$ increases by at most an amplification factor of $\gamma$ in each round, \textit{i.e.}, $\|\delta_i^{(t+1)}\| \leq \gamma \|\delta_i^{(t)}\|$, then by recursively expanding, we obtain the accumulated skew at the last round $T$ as
\begin{equation}\label{eq:cumulative}
\begin{split}
    \|\Delta_{cumulative}^{(u,T)}\| &\leq \eta \sum_{t=0}^{T-1}\sum_{i=1}^{m} \left( C_1 \|\delta_i^{(t)}\| + C_2 |\mathbf{p}_u^{(0)\top} \delta_i^{(0)}| \|\mathbf{q}_i^{(u,0)}\| \right)\\
    & = \sum_{i=1}^{m}\left(\eta C_1 \|\delta_i^{(0)}\| \frac{1 - \gamma^T}{1 - \gamma} + \eta C_2 |\mathbf{p}_u^{(0)\top} \delta_i^{(0)}| \|\mathbf{q}_i^{(u, 0)}\| \frac{1 - \gamma^T}{1 - \gamma}\right).
\end{split}
\end{equation}
Here, we adopt a worst-case perspective, where the term $|\mathbf{p}_u^{(t)\top} \delta_i^{(t)}| \|\mathbf{q}_i^{(u, t)}\|$is upper bounded by its initial value at $t=0$. This reflects the intuition that the initial stage of training is most susceptible to uncontrolled deviation, and thus governs the peak of accumulated skew. Based on the convergence property of FedAvg~\cite{mcmahan2017communication, li2019convergence}, we have $0<\gamma < 1$, thus the upper bound of the accumulated skew remains finite, that is,
\begin{equation}
    \|\Delta_{cumulative}^{(u,T)}\| \leq \sum_{i=1}^{m}\left(\frac{\eta C_1 \|\delta_i^{(0)}\|}{1 - \gamma} + \frac{\eta C_2 |\mathbf{p}_u^{(0)\top} \delta_i^{(0)}| \|\mathbf{q}_i^{(u, 0)}\|}{1 - \gamma}\right).
\end{equation}
% However, when $\gamma \approx 1$, \textit{i.e.}, the error decay is very slow, the error may grow unbounded, leading to convergence failure.
Hence, the proof is complete.\end{proof}

\begin{remark}[Impact of user embedding skew]\label{rmk:impact_skew}
In federated recommendation, the user embedding $\mathbf{p}_u$ is treated as a local variable to preserve personalization. Ideally, it should be jointly optimized with the local item embeddings $\mathbf{q}_i^{(u)}$. However, global aggregation introduces inconsistency to item embeddings, where $\mathbf{q}_i^{(g)} = \mathbf{q}_i^{(u)} + \delta_i$. Consequently, the user embedding $\mathbf{p}_u$ is updated under the influence of aggregated item embeddings, which may differ from the individual preferences of client $u$. Formally, the expected update of the user embedding should be
\begin{equation}
    \mathbf{p}_u^{(t+1)} = \mathbf{p}_u^{(t)} - \eta \nabla_{\mathbf{p}_u} \mathcal{L}(\mathbf{p}_u),
\end{equation}
where $\mathcal{L}(\cdot)$ denotes the local optimization objective, as defined in Equation~(\ref{eq:bce}). Yet due to the introduced skew, the actual update becomes
\begin{equation}
    \mathbf{p}_u^{(t+1)} = \mathbf{p}_u^{(t)} - \eta \nabla_{\mathbf{p}_u} \mathcal{L}(\mathbf{p}_u) + \Delta_{cumulative}^{(u,t)}.
\end{equation}
Ideally, the user embedding should converge to the client-optimal point $\mathbf{p}_u^*$ such that $\nabla_{\mathbf{p}_u} \mathcal{L}(\mathbf{p}_u^*) = 0$.
However, due to the skew accumulated over global rounds, the embedding $\mathbf{p}_u^{(T)}$ generally fails to satisfy this condition. Instead, we observe that $\nabla_{\mathbf{p}_u} \mathcal{L}(\mathbf{p}_u^{(T)}) \approx -\Delta{\text{cumulative}}^{(u,T)} \neq 0$, indicating a suboptimal convergence point under collaborative aggregation. The norm $\|\Delta_{cumulative}^{(u,T)}\|$, which quantifies this deviation, is formally bounded in Lemma~\ref{thm:skew_bound_}. Therefore, it is crucial to mitigate the influence of $\delta_i^{(0)}$ on user embeddings from the very beginning of training, in order to avoid convergence to suboptimal solutions.
\end{remark}

\subsection{User embedding calibration}\label{app:user_calibration}
To mitigate the user embedding skew analyzed in Lemma~\ref{thm:skew_quant}, we first freeze the item embeddings $\mathbf{Q}_{g}$ to keep their influence on the user embedding $\mathbf{p}_u$ controllable. Then, we inject a learnable matrix $\mathbf{W}_{u}\in\mathbb{R}^{m\times d}$ into the frozen item embeddings $\mathbf{Q}_{u}$ as a buffer for local calibration. 
\begin{lemma}[Calibration of user embedding skew]
\label{thm:skew_calibration}
For global round $t$, the user embedding calibration of client $u$ achieved by our method can be approximated as
\begin{equation}
    \Delta_u^{(t){\prime}}\approx-\sum_{i=1}^{m}\left( \underbrace{\eta L_1^{(t)} L_2^{(t)} \mathbf{p}_u^{(t)\top} \mathbf{p}_u^{(t)} \cdot \mathbf{q}_i^{(u,t)}}_{\text{Scaling Term}} + \underbrace{\eta L_1^{(t)2} \cdot \mathbf{p}_u^{(t)}}_{\text{Shift Term}} \right),
\end{equation}
where $\eta$ is the local learning rate for embedding.
\end{lemma}

\begin{proof}The calibration gradients are formulated as follows,
\begin{equation} 
\begin{split}
    \frac{\partial \mathcal{L}_{u}}{\partial \mathbf{w}_i} = &( \sigma(\mathbf{p}_u^\top(\mathbf{q}_i^{(g)} + \mathbf{w}_i)) - r_{ui} )\mathbf{p}_u, \\
    \frac{\partial \mathcal{L}_u}{\partial \mathbf{p}_u} = &(\sigma(\mathbf{p}_u^\top (\mathbf{q}_i^{(g)} + \mathbf{w}_i)) - r_{ui}) \mathbf{q}_i^{(g)}.
\end{split}
\end{equation}
where $\mathbf{w}_i$ denotes the $i$-th row vector of the matrix $\mathbf{W}_u$. Since $\mathbf{W}_{u}$ is zero-initialized, similar to the proof of Lemma~\ref{thm:skew_quant}, we have
\begin{equation}
\begin{split}
    \frac{\partial \mathcal{L}_{u}}{\partial \mathbf{w}_i} = & (\sigma(\mathbf{p}_u^\top \mathbf{q}_i^{(u)}) - r_{ui}) \mathbf{p}_u +\sigma^{\prime}(\mathbf{p}_u^\top\mathbf{q}_i^{(u)})\mathbf{p}_u^\top (\delta_i+\mathbf{w}_i)\mathbf{p}_u  +\mathcal{O}(||\delta_i+\mathbf{w}_i||^2), \\
    \frac{\partial \mathcal{L}_u}{\partial \mathbf{p}_u} = & (\sigma(\mathbf{p}_u^\top \mathbf{q}_i^{(u)}) - r_{ui}) \mathbf{q}_i^{(u)} +\sigma^{\prime}(\mathbf{p}_u^\top\mathbf{q}_i^{(u)})\mathbf{p}_u^\top (\delta_i+\mathbf{w}_i)\mathbf{q}_i^{(u)} + (\sigma(\mathbf{p}_u^\top \mathbf{q}_i^{(u)}) - r_{ui})(\delta_i+\mathbf{w}_i)\\
    +&\mathcal{O}(||\delta_i+\mathbf{w}_i||^2).
\end{split}
\end{equation}
For simplicity, we ignore the higher-order infinitesimal terms and denote $L_1 = \sigma(\mathbf{p}_u^\top \mathbf{q}_i^{(u)}) - r_{ui}$, $L_2 = \sigma^{\prime}(\mathbf{p}_u^\top\mathbf{q}_i^{(u)})$ and $ \mathbf{v}_i = \delta_i + \mathbf{w}_i$, then we have
\begin{equation}
    \frac{\partial \mathcal{L}_{u}}{\partial \mathbf{w}_i} = L_1 \mathbf{p}_u +L_2\mathbf{p}_u^\top \mathbf{v}_i\mathbf{p}_u
\end{equation}
\begin{equation}\label{eq:p_u_L}
    \frac{\partial \mathcal{L}_u}{\partial \mathbf{p}_u} = L_1 \mathbf{q}_i^{(u)} +L_2\mathbf{p}_u^\top \mathbf{v}_i\mathbf{q}_i^{(u)} + L_1\mathbf{v}_i.
\end{equation}
According to Remark~\ref{rmk:impact_skew}, we perform calibration from the beginning of training to avoid suboptimal convergence. Since $\mathbf{W}$ is zero-initialized, the updated buffer matrix can be written as $\mathbf{w}_i = -\eta ( L_1 \mathbf{p}_u + L_2 \mathbf{p}_u^\top \mathbf{v}_i \mathbf{p}_u )$. Therefore, we have
\begin{equation}\label{eq:v_i}
    \mathbf{v}_i=\delta_i + \mathbf{w}_i=\delta_i-\eta(L_1 \mathbf{p}_u +L_2\mathbf{p}_u^\top \mathbf{v}_i\mathbf{p}_u),
\end{equation}
where $\eta$ denotes the local learning rate. Since $\eta \leq 0.01$ is sufficiently small, and $L_1 \in (0, 1)$, $L_2 \in (0, 0.25]$, we apply a first-order approximation to Equation~(\ref{eq:v_i}), yielding $\mathbf{v}_i \approx \delta_i - \eta L_1 \mathbf{p}_u$. Thus Equation~(\ref{eq:p_u_L}) can be rewritten as
\begin{equation}
\begin{split} 
    \frac{\partial \mathcal{L}_u}{\partial \mathbf{p}_u} \approx & L_1 \mathbf{q}_i^{(u)} + L_2 (\mathbf{p}_u^\top (\delta_i - \eta L_1 \mathbf{p}_u)) \mathbf{q}_i^{(u)} + L_1 (\delta_i - \eta L_1 \mathbf{p}_u)\\
    = &\underbrace{L_1 \cdot \mathbf{q}_i^{(u)}}_{\text{Gradient}} + \underbrace{(L_2 \mathbf{p}_u^\top \delta_i \cdot \mathbf{q}_i^{(u)} + L_1 \cdot \delta_i)}_{\text{User Embedding Skew}} - \underbrace{( \eta L_1 L_2 \mathbf{p}_u^\top \mathbf{p}_u \cdot \mathbf{q}_i^{(u)} + \eta L_1^2 \cdot \mathbf{p}_u )}_{\text{User Embedding Calibration}},
\end{split}
\end{equation}
where the first term denotes the original gradient of user embedding, and the second term is the skew introduced by global aggregation as analyzed in Lemma~\ref{thm:skew_quant}. Our method introduces a third calibration term to mitigate the skew. Hence, the calibration introduced by our method can be approximated as
\begin{equation}
    \Delta_u^{{\prime}}\approx-\sum_{i=1}^{m}\left( \underbrace{\eta L_1 L_2 \mathbf{p}_u^\top \mathbf{p}_u \cdot \mathbf{q}_i^{(u)}}_{\text{Scaling Term}} + \underbrace{\eta L_1^2 \cdot \mathbf{p}_u}_{\text{Shift Term}} \right),
\end{equation}
where the scaling term mitigates the amplification of gradient magnitude, and the shift term reduces the amplification of directional deviation. Together, they compensate for the user embedding skew $\Delta_u$ in Lemma~\ref{thm:skew_quant}, thus theoretically achieving calibration.\\
Hence, the proof is complete.\end{proof}
\begin{remark}[Dynamic regularization perspective for calibration]
Taking the $i$-th item for example, we consider the shift term in Lemma~\ref{thm:skew_calibration} used to correct the gradient direction, \textit{i.e.}, $-\eta L_1^{(t)2} \mathbf{p}_u^{(t)}$. This term is structurally similar to the gradient of $l_2$ regularization term:
\begin{equation}
    \frac{\partial}{\partial \mathbf{p}_u} \left( \frac{1}{2} \|\mathbf{p}_u^{(t)}\|^2 \right) = \mathbf{p}_u^{(t)}.
\end{equation}
Hence, the shift term can be interpreted as an adaptive regularization term applied to the user embedding, with a coefficient of $-\eta L_1^{(t)2}$ that dynamically adjusts according to the user embedding skew at each round $t$. In this way, it serves as a self-regulating force that suppresses harmful gradient directions induced by global aggregation, offering robust user embedding updates.
\end{remark}

\subsection{Cost analysis}\label{app:cost}

\begin{table}[thb]
    \centering
    \setlength{\abovecaptionskip}{2pt}
    \caption{Cost analysis of PFedCLR against other SOTA FR methods. "$\uparrow$" denotes increased costs on FedMF, while "$-$" indicates no significant additional cost.}
    \label{tab:cost}
    \resizebox{1.0\textwidth}{!}{
    \begin{tabular}{lcccccc}
    \toprule
    \multirow{2}{*}{\textbf{Method}}& \multicolumn{2}{c}{\textbf{Client-side Cost}} & \multicolumn{2}{c}{\textbf{Server-side Cost}} & \multicolumn{2}{c}{\textbf{Round Efficiency}} \\
    \cmidrule(lr){2-3}\cmidrule(lr){4-5}\cmidrule(lr){6-7}
    & \small\textbf{Time} & \small\textbf{Space} & \small\textbf{Time} & \small\textbf{Space} & \small\textbf{Communication} & \small\textbf{Computation} \\
    \midrule
    FedMF & $O(k(m+1)d)$ & $O((m+1)d)$ & $O(nmd)$ & $O((n+1)md)$ & $O(2nmd)$ & $O(k^{max}(m+1)d+nmd)$ \\
    FedNCF & $\uparrow O(kLd^2)$ & $\uparrow O(Ld^2)$ & $\uparrow O(nLd^2)$ & $\uparrow O((n+1)Ld^2)$ & $\uparrow O(2nLd^2)$ & $\uparrow O(k^{max}Ld^2+nLd^2)$ \\
    pFedGraph & - & - & $\uparrow O(n^4+n^2md+n^2md)$ & $\uparrow O(n^2+n^2md)$ & - & $\uparrow O(n^4 + n^2md+n^2md)$ \\
    GPFedRec & $\uparrow O(kLd^2)$ & $\uparrow O(Ld^2)$ & $\uparrow O(n^2md+n^2md)$ & $\uparrow O(n^2+n^2md)$ & - & $\uparrow O(n^2md+n^2md)$ \\
    PFedRec & $\uparrow O(kLd^2)$ & $\uparrow O(Ld^2)$ & - & - & - & $\uparrow O(k^{max}Ld^2)$ \\
    FedRAP & $\uparrow O(kmd)$  & $\uparrow O(md)$ & - & - & - & $\uparrow O(k^{max}md)$  \\
    PFedCLR & $\uparrow O(k(m+d)r)$ & $\uparrow O((m+d)r)$ & - & - & - & - OR $\uparrow O(k^{max}mr-nmd)$ \\
    \bottomrule
    \end{tabular}
    }
\end{table}

Since most FR methods, including ours, are embedding-based, we conduct the cost analysis using the classic FedMF as the benchmark. We analyze the time and space complexity from both the client and server sides. Moreover, we evaluate the efficiency of each global round, including the communication cost, \textit{i.e.}, the total size of transmitted model parameters, and the computation cost, \textit{i.e.}, the maximum local training time across all clients plus the global aggregation time. The costs of PFedCLR against other SOTA FR methods are summarized in Table~\ref{tab:cost}, where $n$ and $m$ represent the number of clients and items, respectively. Additionally, $k$ denotes the number of local interactions of each client, and $k^{\text{max}}$ is the maximum among all clients. Besides, $d$ denotes the dimension of embedding, and $L$ denotes the number of model layers for methods with MLP or other networks.

Concretely, FedNCF extends FedMF by adding an $L$-layer MLP on each client and requires the server to aggregate both item embeddings and MLP parameters. Besides, global aggregation-based pFR methods incur additional server-side overhead by deriving client-specific global models. In particular, pFedGraph solves a convex optimization problem for each client, resulting in high computational costs and making it impractical to sample all clients for aggregation in each round. GPFedRec also builds on FedNCF but omits the aggregation of MLPs on the server. In contrast, local adaptation-based methods mainly introduce local model overhead. For example, PFedRec incorporates an $L$-layer score function to model user-specific preferences, while FedRAP adds an additive item embedding table on each client to enhance personalization.

Notably, PFedCLR incurs no additional server-side overhead and only adds the low-rank matrices of size $\mathcal{O}((m+d)r)$ on each client. Given that $r \ll \min(m,d)$ in practice, this overhead is negligible compared to the backbone model size $\mathcal{O}((m+1)d)$. Moreover, PFedCLR does not increase communication cost per round. Different from other FR methods, where client training and server aggregation are executed alternately, PFedCLR enables parallel computation of client updates and server aggregation, \textit{i.e.}, Step 2 and Step 3, further boosting training efficiency.

\subsection{Privacy analysis}\label{app:privacy}
Although our method follows the federated paradigm, ensuring private data remains locally without being uploaded to the server~\cite{sun2024survey, yin2024device}, existing studies have highlighted the potential risk of inferring local data distributions from uploaded models~\cite{chai2020secure, wu2022federated}. To further enhance privacy protection, we integrate the $\varepsilon$-differential privacy technique with our PFedCLR, where $\varepsilon$ represents the privacy budget, measuring the level of privacy protection. A smaller $\varepsilon$ implies stronger privacy. Specifically, we achieve local differential privacy (LDP) by adding the zero-mean Laplace noises to the uploaded local model $\mathbf{Q}_u$~\cite{dwork2006calibrating, choi2018guaranteeing}, as follows,
\begin{equation}
    \mathbf{Q}_u^{\prime} = \mathbf{Q}_u + \text{Laplace}(0, \lambda),
\end{equation}
where $\lambda=\mathcal{S}_u/\varepsilon$ is the scale parameter of the Laplace distribution~\cite{dwork2006calibrating}. $\mathcal{S}_u$ denotes the global sensitivity of Client $u$, with its upper bound provided in Lemma~\ref{thm:ldp}. Thus, a larger $\lambda$ indicates stronger privacy protection.
\begin{lemma}[Bound of global sensitivity $\mathcal{S}_u$]\label{thm:ldp}
Given two global $\mathbf{Q}_g$ and $\mathbf{Q}_g^{\prime}$, learned from two datasets differing only in the data of client $u$, \textit{i.e.}, $\mathcal{D}_u$ and $\mathcal{D}_u^{\prime}$, respectively. We have
\begin{equation}
    \mathcal{S}_u= || \mathbf{Q}_g - \mathbf{Q}_g^{\prime} || = || p_u \eta (\nabla \mathbf{Q}_u^{\mathcal{D}_u} - \nabla \mathbf{Q}_u^{\mathcal{D}_u^{\prime}}) || \leq  p_u \eta (||(\nabla \mathbf{Q}_u^{\mathcal{D}_u}|| + ||\nabla \mathbf{Q}_u^{\mathcal{D}_u^{\prime}}) ||) \leq 2p_u\eta C,
\end{equation}
where $p_u$ denotes the weight assigned to $\mathbf{Q}_u$ for global aggregation and $\eta$ denotes the local learning rate for embedding. Besides, $\nabla \mathbf{Q}_u^{\mathcal{D}}$ denotes the gradient of the learned model by client $u$ on dataset $\mathcal{D}$, which can be constrained by a clipping technique to ensure $\|\nabla \mathbf{Q}_u^{\mathcal{D}}\| \leq C$, where $C$ is a constant~\cite{wei2020federated}.
\end{lemma}

\section{Experimental Details}\label{app:exp_details}
\subsection{Dataset details}\label{app:datasets}

\begin{table}[thb]\centering
    \setlength{\abovecaptionskip}{2pt}
    \caption{Statistics of the experimental datasets.}
    \label{tab:dataset_statistic}
    \resizebox{0.55\textwidth}{!}{
    \large
    \begin{tabular}{*{5}{c}}
        \toprule
       Dataset & \#User/Client & \#Item & \#Interaction & Sparsity \\
        \midrule
        Filmtrust & 1,002 & 2,042 & 33,372 & 98.37\% \\
        ML-100K & 943 & 1,682 & 100,000 & 93.70\% \\
        ML-1M & 6,040 & 3,706 & 1,000,209 & 95.53\% \\
        Hetrec2011 & 2,113 & 10,109 & 855,598 & 95.99\% \\
        LastFm-2K & 1,600 & 12,454 & 185,650 & 99.07\% \\
        \bottomrule
    \end{tabular}
    }
\end{table}

We conduct extensive experiments on five recommendation benchmark datasets with varying scales and sparsity: \textbf{Filmtrust}~\cite{filmtrust_2013}, \textbf{Movielens-100K (ML-100K)}~\cite{harper2015movielens}, \textbf{Movielens-1M (ML-1M)}~\cite{harper2015movielens}, \textbf{HetRec2011}~\cite{cantador2011second} and \textbf{LastFM-2K}~\cite{cantador2011second}.
The first four datasets are for movie recommendation with user-movie ratings, the last dataset is for music recommendation with user-artist listening count. We convert ratings/counts greater than 0 to 1, targeting the recommendation task with implicit feedback data. Following previous works~\cite{zhang2023dual, li2023federated}, we filter out the users with less than 10 interactions from the above datasets and treat each user as an independent client. Detailed statistics of the five datasets are summarized in Table~\ref{tab:dataset_statistic}.

\subsection{Baseline details}\label{app:baselines}
We compare several centralized recommender system algorithms as follows:

\textbf{Matrix Factorization (MF)}~\cite{koren2009matrix}: It is a classic recommendation model that factorizes the rating matrix into two embeddings within the same latent space, capturing the characteristics of users and items, respectively.\\
\textbf{Neural Collaborative Filtering (NCF)}~\cite{he2017neural}: It builds upon the learned user and item embeddings and further utilizes an MLP to model the user-item interactions, introducing a high level of non-linearity.\\
\textbf{LightGCN}~\cite{he2020lightgcn}: It is a collaborative filtering model based on a simplified graph convolutional network, which learns user and item embeddings by linearly propagating them on the user-item interaction graph efficiently.

Additionally, for a comprehensive comparison, we select several representative SOTA methods from each of the three types of FR methods, as follows:

I) Traditional FR methods:\\
\textbf{FedMF}~\cite{chai2020secure}: It is the federated version of MF. It trains user embedding locally and uploads encrypted item gradients to the server for global aggregation. For the purpose of performance evaluation, we adopt its unencrypted version.\\
\textbf{FedNCF}~\cite{perifanis2022federated}: It is the federated version of NCF. It updates user embedding locally and uploads both item embeddings and MLP to the server for global aggregation.

II) Global aggregation-based pFR methods:\\
\textbf{pFedGraph}~\cite{ye2023personalized}: It formulates the fine-grained optimization based on similarity, aggregating a unique global model for each client. Besides, it optimizes local models based on aggregated models on the
client side to facilitate personalization. This method, as a general federated learning approach, has been proven effective in recommendation tasks.\\
\textbf{GPFedRec}~\cite{zhang2024gpfedrec}: It introduces a graph-guided aggregation mechanism that facilitates the learning of user-specific item embeddings globally, thereby promoting user personalization modeling.

III) Local adaptation-based pFR methods:\\
\textbf{PFedRec}~\cite{zhang2023dual}: It introduces a novel dual personalization mechanism to capture user preferences by a score function and obtain fine-grained item embeddings.\\
\textbf{FedRAP}~\cite{li2023federated}: It balances global shared and local personalized knowledge by applying an additive model to item embeddings, enhancing the recommendation performance.

\subsection{Implementation details}\label{app:implementation}
In the experiment, we randomly sample $N=4$ negative instances for each positive sample following previous works~\cite{he2017neural,zhang2023dual}. Considering the fairness of comparison, we set the embedding dimension $d=16$ and batch size $B=256$ for all methods, and other baseline details are followed from the original paper. For the centralized methods, we set the total epoch as 100 to guarantee their convergence. For the federated methods, we set the global round $R=100$ for generality. Notably, for pFedGraph, we perform personalized aggregation for only $10\%$ of the clients participating in each round. This is because recommendation tasks typically involve tens of thousands of clients, and optimizing a personalized model for each client on the server is computationally expensive, which would incur significantly higher costs than other methods.

For our method, we follow the same basic hyper-parameter settings as the backbone FedMF, that is, sampling $s=60\%$ of clients for global aggregation per round, the learning rate $\eta=0.01$ for user and item embeddings, and $E=10$ local training epochs using the Adam optimizer. As for the additional hyper-parameters introduced by PFedCLR, we search the rank $r$ of low-rank matrices in $\{1, 2, 4, 6, 8, 10, 12\}$, and the learning rate $\beta$ in $\{0.1, 0.01, 0.001, 0.0001\}$ via the validation set performance. In this work, we conduct the experiments on a GPU server with NVIDIA RTX A5000 and report the results as the average of 5 repeated experiments.

\section{Extensive Experiment Results}\label{app:more_experiment}
\subsection{Motivation}\label{app:pre_exp}

\begin{figure*}[htbp]
\centering
\includegraphics[width=1.0\textwidth]{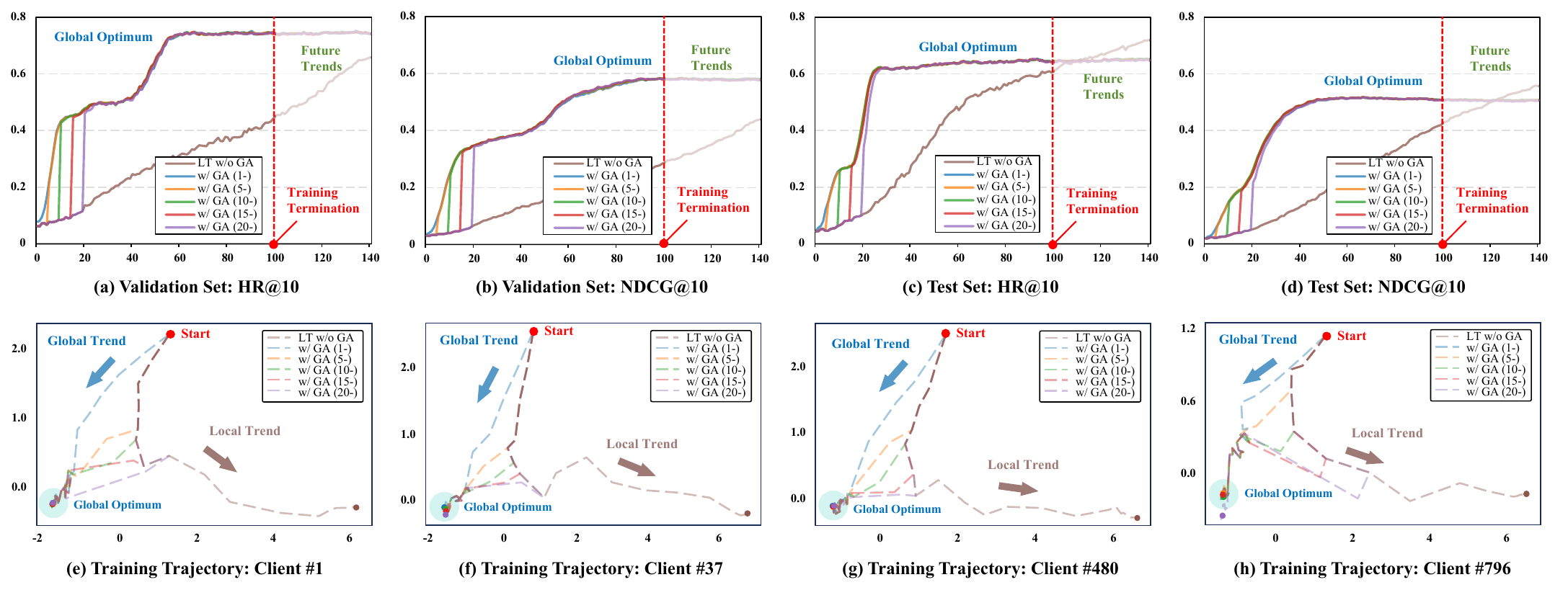} % Reduce the figure size so that it is slightly narrower than the column.
\caption{More pre-experimental results for motivation. Subfigures (a) and (b) show performance on the validation set, while (c) and (d) report results on the test set. Subfigures (e) to (h) illustrate the training trajectories of user embeddings for different clients. In the legend, "LT w/o GA" denotes local training without global aggregation, and "w/ GA ($i$-)" indicates that global aggregation is introduced starting from round $i$, rather than local training only.}
\label{pic:more_preexp}
\end{figure*}

\textbf{Experimental setup.} We use FedMF as the backbone and conduct the experiment on Filmtrust dataset. For a fair comparison, the parameter settings are uniform for all experimental groups as follows: local epochs $E=10$, global rounds $R=100$, and the additional 40 rounds are performed to explore the future trend with and without global aggregation. The learning rate of user embeddings and item embeddings $\eta=0.01$, batch size $B=256$, and embedding dimension $d=16$. Besides, we randomly sample $N=4$ negative instances for each positive sample following previous works~\cite{he2017neural, zhang2023dual}. In addition, the model performance is reported by HR@10 and NDCG@10.

\textbf{Experimental results.} More pre-experimental results for motivation are shown in Figure~\ref{pic:more_preexp}. The observation described in Section~\ref{sec:intro} might raise concerns that the superior performance under local training is due to overfitting. To address this, we report results on both the validation and test sets in subfigures (a) to (d). The consistent trends across both sets confirm that our motivation is well-founded rather than a consequence of overfitting. Focusing on test performance under the standard setting, \textit{i.e.}, training is terminated after 100 global rounds, groups with global aggregation converge to a global optimum, while the one with only local training converges more slowly and perform worse. However, from the future trend, the group without global aggregation can achieve better performance. These observations above suggest that global aggregation accelerates convergence but leads to suboptimal results.

Given the specificity of the FR scenarios, \textit{i.e.}, each client retains local user embedding without engaging in global aggregation, we further visualize the training trajectories of user embeddings. As some cases shown in subfigures (e) to (h) of Figure~\ref{pic:more_preexp}, user embeddings exhibit a consistent updating trend once global aggregation begins, which differs significantly from the trend under local training only. This indicates that globally aggregated item embeddings influence the optimization direction of local user embeddings, causing clients to converge toward a collaborative but suboptimal solution, \textit{i.e.}, the user embedding skew issue. To this end, our method not only personalizes item embeddings to enable more accurate recommendations but also calibrates user embeddings to mitigate the skew introduced by global aggregation.

\subsection{Efficiency}\label{app:efficiency}

\begin{table}[htbp]
\caption{Space comparison of different federated methods (unit: MB). "Overhead" indicates the additional cost of PFedCLR over the backbone FedMF.}
\centering
\resizebox{\textwidth}{!}{
\begin{tabular}{lcccccccccc}
\toprule
\multirow{2}{*}{\textbf{Method}}& \multicolumn{2}{c}{\textbf{Filmtrust}} & \multicolumn{2}{c}{\textbf{ML-100K}} & \multicolumn{2}{c}{\textbf{ML-1M}} & \multicolumn{2}{c}{\textbf{HetRec2011}} & \multicolumn{2}{c}{\textbf{LastFM-2K}} \\
\cmidrule(lr){2-3}\cmidrule(lr){4-5}\cmidrule(lr){6-7}\cmidrule(lr){8-9}\cmidrule(lr){10-11}
&\small\textbf{Client} & \small\textbf{Server} & \small\textbf{Client} & \small\textbf{Server} & \small\textbf{Client} & \small\textbf{Server} & \small\textbf{Client} & \small\textbf{Server} & \small\textbf{Client} & \small\textbf{Server} \\
\midrule
FedMF & 0.1247 & 125.0077 & 0.1027 & 96.9121 & 0.2263 & 1366.4518 & 0.6171 & 1304.3473 & 0.7521 & 955.1355 \\
FedNCF & 0.1273 & 127.5827 & 0.1053 & 99.3356 & 0.2288 & 1381.9608 & 0.6196 & 1309.7745 & 0.7547 & 958.3960 \\
pFedGraph & 0.1247 & 137.3846 & 0.1027 & 106.4933 & 0.2263 & 1504.2398 & 0.6171 & 1434.0880 & 0.7521 & 1049.2055 \\
GPFedRec & 0.1273 & 253.7207 & 0.1053 & 197.1138 & 0.2288 & 2871.8436 & 0.6196 & 2625.1093 & 0.7547 & 1915.6620 \\
PFedRec & 0.1247 & 125.0115 & 0.1027 & 96.9157 & 0.2263 & 1366.4748 & 0.6171 & 1304.3554 & 0.7521 & 955.1403 \\
FedRAP & 0.2493 & 125.0077 & 0.2054 & 96.9121 & 0.4525 & 1366.4518 & 1.2341 & 1304.3473 & 1.5042 & 955.1355 \\
PFedCLR & 0.1404 & 125.0077 & 0.1157 & 96.9121 & 0.2547 & 1366.4518 & 0.6943 & 1304.3473 & 0.8463 & 955.1355 \\
\midrule\midrule
Overhead & +0.0157 & - & +0.0130 & - & +0.0284 & - & +0.0772 & - & +0.0942 & - \\
\bottomrule
\end{tabular}
}
\label{tab:exp_space}
\end{table}

\begin{table}[thb]\centering
    \setlength{\abovecaptionskip}{2pt}
    \caption{Time comparison of different federated methods (unit: s). The data represents the total time spent per round, \textit{i.e.}, the maximum time of local updating over all clients together with the time of global aggregation by the server, ignoring the communication time.}
    \label{tab:exp_time}
    \resizebox{0.6\textwidth}{!}{
    \large
    \begin{tabular}{lccccc}
        \toprule
        \textbf{Method} & \textbf{Filmtrust} & \textbf{ML-100K} & \textbf{ML-1M} & \textbf{HetRec2011} & \textbf{LastFM-2K} \\
        \midrule
        FedMF & 0.5069 & 0.4831 & 0.7074 & 0.7425 & 0.5455 \\
        FedNCF & 0.7099 & 0.6682 & 1.4821 & 1.0285 & 0.7360 \\
        pFedGraph & 5.6346 & 5.1800 & 185.9420 & 25.3069 & 10.1281 \\
        GPFedRec & 3.0124 & 2.4704 & 32.5816 & 10.2724 & 6.9336 \\
        PFedRec & 1.2848 & 1.4508 & 3.5830 & 2.6183 & 2.0975 \\
        FedRAP & 0.6217 & 0.6354 & 4.1745 & 2.7195 & 1.7354 \\
        PFedCLR & 0.6337 & 0.6689 & 1.2234 & 1.7539 & 1.0926 \\
        \bottomrule
    \end{tabular}
    }
\end{table}

We compare the space overhead of different federated methods on both client and server sides, as summarized in Table~\ref{tab:exp_space}. Global aggregation-based methods \textit{e.g.}, pFedGraph and GPFedRec, incur additional server-side cost, as they need to maintain a personalized global model for each client besides the uploaded local models. On the other hand, local adaptation-based methods introduce extra overhead on the client side. In particular, FedRAP incurs nearly double the overhead compared to other methods, which is unfavorable for resource-constrained devices. In contrast, our method introduces only a slight overhead on the client compared to the most lightweight baseline FedMF, while achieving a notable performance improvement.

We further compare the per-round training time of different federated methods. As shown in Table~\ref{tab:exp_time}, global aggregation-based pFR methods, such as pFedGraph and GPFedRec, incur longer training time. This is primarily because clients must wait for server-side aggregation to complete before proceeding with local updates. Meanwhile, the server has to derive personalized models for each client, resulting in increased computation time and forming an efficiency bottleneck. Unlike existing methods where global aggregation and local updates are performed alternately, PFedCLR allows them to proceed in parallel. This concurrent execution eliminates idle waiting on the client side, improves training efficiency, and reduces overall training time.

\begin{figure*}[htbp]
\centering
\includegraphics[width=1.0\textwidth]{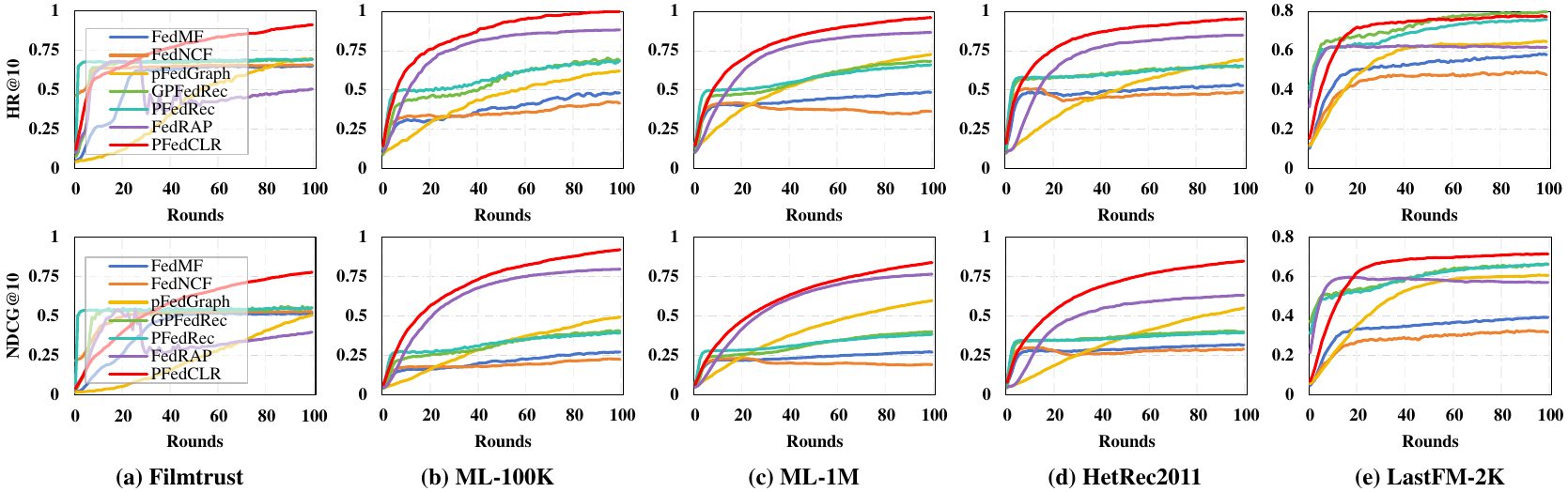} % Reduce the figure size so that it is slightly narrower than the column.
\caption{Model convergence comparison. The horizontal axis is the federated rounds, and the vertical axis is the model performance on both HR@10 and NDCG@10.}
\label{pic:convergence}
\end{figure*}

\subsection{Convergence}\label{app:convergence}
We compare the convergence of our method with baselines under two metrics, as illustrated in Figure~\ref{pic:convergence}. We provide the following analysis:\\
I) Early-stage convergence: PFedCLR exhibits a faster convergence speed in the early training rounds and consistently outperforms other methods on most datasets, particularly on ML-100K, ML-1M, and HetRec2011. This advantage stems from its dual-function mechanism, which not only personalizes the global item embeddings but also calibrates the user embeddings, enhancing model performance from both perspectives.\\
II) Late-stage convergence: In the later rounds, PFedCLR maintains stable convergence and achieves superior final performance across various datasets. This is attributed to the continuous calibration of local user embeddings throughout training, which prevents them from being skewed by globally aggregated item embeddings and helps avoid convergence to suboptimal global solutions, instead promoting client-specific optimization.\\
In summary, PFedCLR demonstrates both short-term efficiency and long-term effectiveness, highlighting its strong potential in practical federated recommendation scenarios.

\subsection{Trajectory visualization}\label{app:trajectory}

\begin{figure*}[htbp]
\centering
\includegraphics[width=1.0\textwidth]{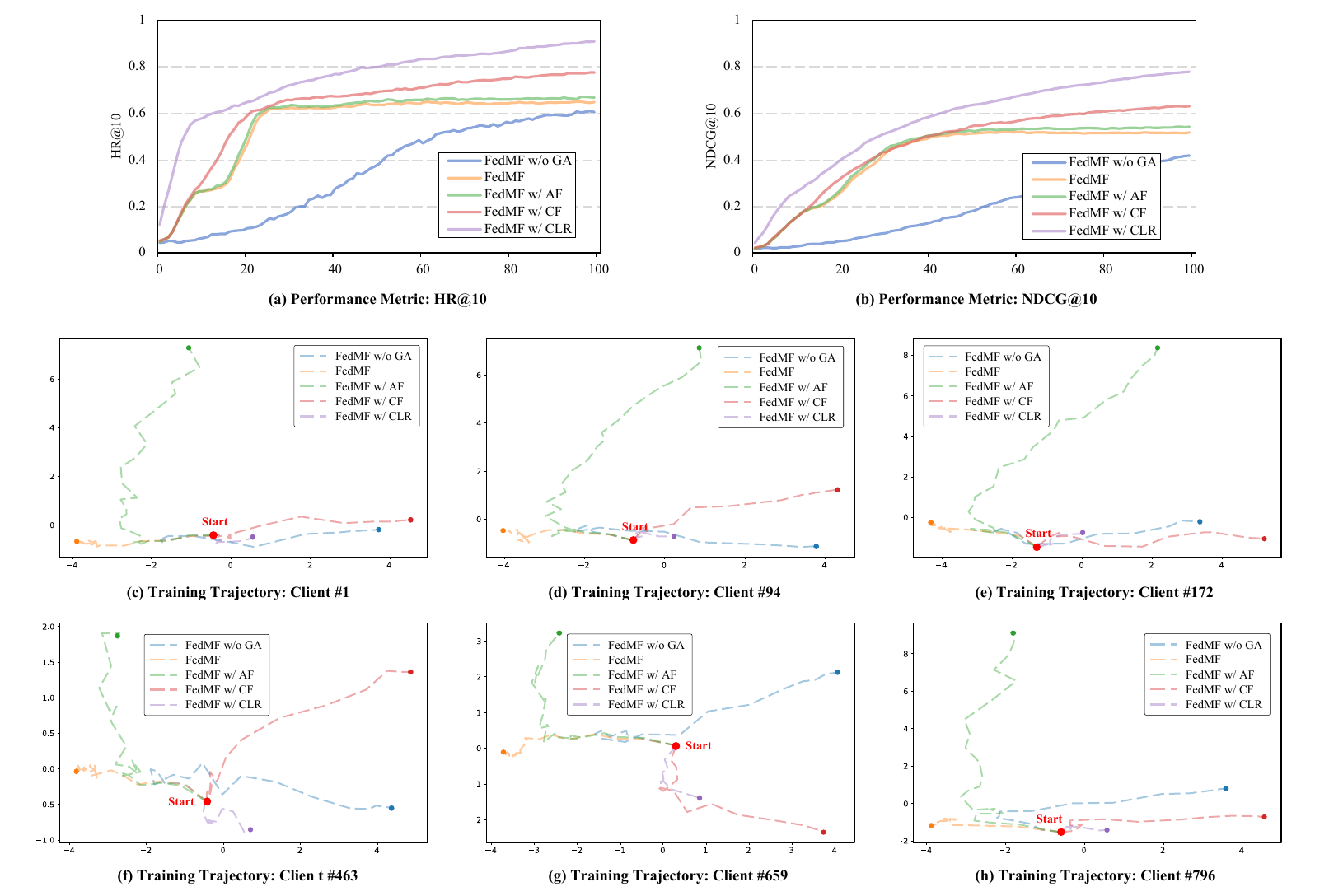} % Reduce the figure size so that it is slightly narrower than the column.
\caption{Performance and training trajectory under different variants. "GA", "AF", "CF", and "CLR" denote Global Aggregation, Adaptation with Full matrix, Calibration with Full matrix, and Calibration with Low-Rank matrices, respectively.}
\label{pic:vis_traj}
\end{figure*}

To further demonstrate the calibration effect of our method on the user embedding, we randomly sample several clients on Filmtrust and visualize the training trajectories of their user embeddings under different variants discussed in Subsection~\ref{sec:ablation} using t-SNE. As shown in Figure~\ref{pic:vis_traj}, we derive the following observations and analysis:\\
I) Global aggregation causes user embedding skew. Comparing "FedMF w/o GA" with "FedMF", we observe that global aggregation significantly alters the trajectory of local updates. While it aligns clients toward a shared global optimum, it suppresses personalized optimality, resulting in user embedding skew and suboptimal performance, as discussed in Appendix~\ref{app:pre_exp}.\\
II) Personalized adaptation improves performance. 	"FedMF w/ AF" outperforms "FedMF" by locally adapting the global model for each client, which facilitates more tailored optimization for user embedding and partially alleviates the uniformity imposed by a shared global model. However, the influence of the global model on user embeddings remains significant, as no explicit calibration mechanism is applied, akin to typical pFR methods.\\
III) Calibration via a buffer matrix mitigates this skew. Comparing "FedMF w/ CF" and "FedMF", injecting a learnable calibration matrix into global model can buffer the influence on user embedding. However, employing a full matrix risks overfitting to global information, potentially leading to over-calibration and degraded performance.\\
IV) Low-rank decomposition can brings further improvements in calibration. "FedMF w/ CLR" outperforms "FedMF w/ CF" by replacing the full matrix with the low-rank ones. Compared to a full matrix, the low-rank matrices impose inherent constraints, serving as a regularizer that prevents excessive suppression of global aggregation. While mitigating user embedding skew, they effectively preserve beneficial global information.\\
In summary, beyond personalization, PFedCLR effectively and efficiently calibrates the user embedding skew introduced by global aggregation with a lightweight low-rank mechanism.

\subsection{Embedding visualization}\label{app:embedding}

\begin{figure*}[htbp]
\centering
\includegraphics[width=1.0\textwidth]{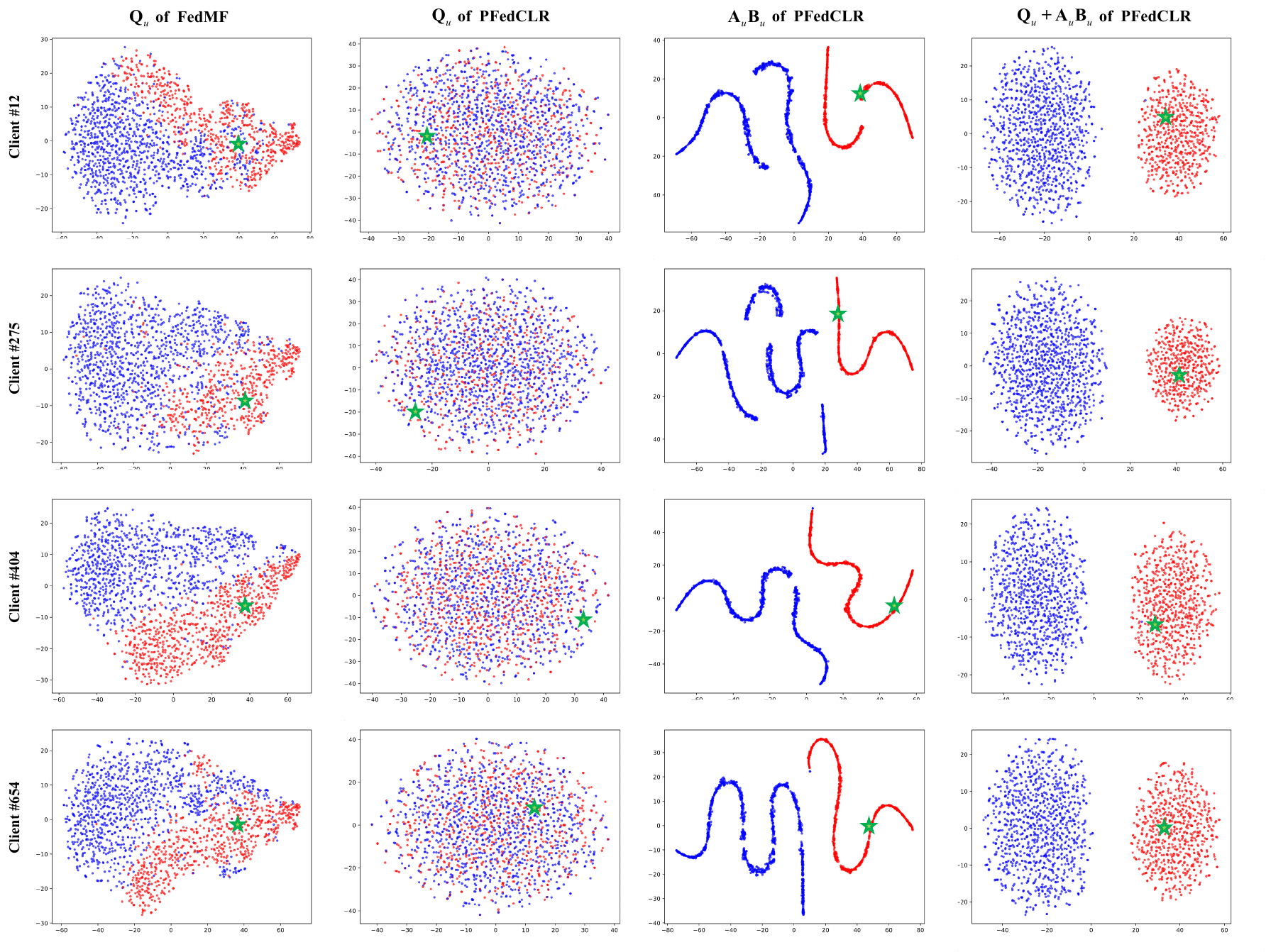} % Reduce the figure size so that it is slightly narrower than the column.
\caption{T-SNE visualizations of user and item embeddings learned by FedMF and PFedCLR. Each row represents a different client, while each column represents a different component. The target \textcolor{green}{client} is indicated by the \textcolor{green}{green} pentagram. Items \textcolor{red}{interacted} with and \textcolor{blue}{not interacted} with by the client are indicated in \textcolor{red}{red} and \textcolor{blue}{blue}, respectively.}
\label{pic:vis_ablation}
\end{figure*}

By applying t-SNE for dimensionality reduction, we visualize the user and item embeddings of FedMF and PFedCLR on ML-100K, with several sampled client cases shown in Figure~\ref{pic:vis_ablation}. We summarize the following observations:\\
I) Calibration (1st–2nd columns): In FedMF, user embeddings are significantly influenced by item embeddings $\mathbf{Q}_u$, appearing closer to the interacted items. As global aggregation alters item embeddings, it further affects user embeddings and leads to user embedding skew. In contrast, PFedCLR decouples the updates of user and item embeddings and employ the zero-initialized $\mathbf{A}_u\mathbf{B}_u$ to suppress mutual influence between the them. As a result, user embeddings exhibit no clear relative alignment with item embeddings, indicating effective calibration.\\
II) Personalization (3rd–4th columns): With the dual-function mechanism, user preference information is further captured by $\mathbf{A}_u\mathbf{B}_u$, which personalizes the item embeddings $\mathbf{Q}_u$, enabling more accurate and client-specific recommendations.

\subsection{Hyper-parameter analysis}\label{app:parameter}

\begin{figure*}[htbp]
\centering
\includegraphics[width=0.9\textwidth]{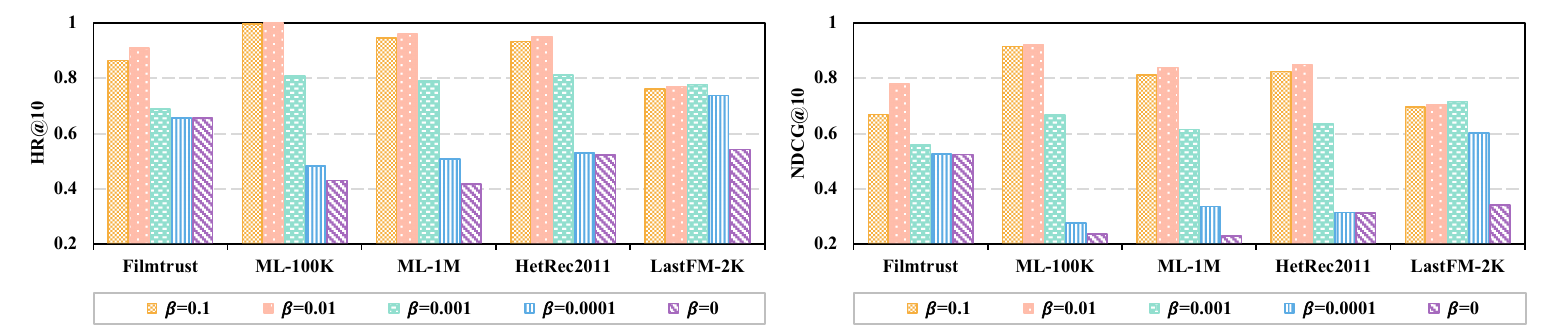} % Reduce the figure size so that it is slightly narrower than the column.
\caption{Model performance under different learning rate $\beta$ of low-rank matrices.}
\label{pic:hyper_beta}
\end{figure*}

Built upon the backbone FedMF, our method only introduces the lightweight low-rank matrices on the client side, requiring just two additional hyper-parameters. All other settings remain consistent with FedMF. We analyze the sensitivity of the two hyper-parameters as follows:\\
I) Rank $r$ of the low-rank matrices.
As shown in Figure~\ref{pic:hyper_rank}, we vary $r$ within $\{1, 2, 4, 6, 8, 10, 12\}$. Increasing $r$ does not lead to performance improvement and may even cause slight degradation, likely because it excessively suppresses the influence of the global model, thereby diminishing the beneficial collaborative information. A small value such as $r=2$ achieves consistently strong results across all datasets, introducing only $\sim10\%$ additional client-side overhead and offering a favorable trade-off between performance and efficiency.\\
II) Learning rate $\beta$ for the low-rank matrices.
Figure~\ref{pic:hyper_beta} shows performance with $\beta$ chosen from $\{0.1, 0.01, 0.001, 0.0001, 0\}$. Extremely small or large $\beta$ leads to suboptimal results. When $\beta = 0$, PFedCLR degenerates to FedMF. The best performance is achieved when $\beta$ matches the embedding learning rate $\eta = 0.01$.

\subsection{Privacy protection}\label{app:privacy_exp}
I) Inherent privacy preservation. We visualize the uploaded item embeddings by different methods on ML-100K with t-SNE, as illustrated in Figure~\ref{pic:vis_item}. Traditional FR methods, such as FedMF and FedNCF, do not personalize local models, and thus, the uploaded embeddings reveal little about user-specific interaction patterns. In contrast, pFR methods like pFedGraph, GPFedRec, and PFedRec explicitly personalize the local models, leading to the clear distinctions between interacted and non-interacted items. While this improves performance, it also increases the risk of user preference leakage. Besides, FedRAP maintains two separate models on the client side, \textit{i.e.}, a globally shared model and a user-specific model, to decouple global and personalized information. Although only the shared model is uploaded to the server, we observe that for certain clients, \textit{e.g.}, client \#449, the uploaded embeddings still reveal user-specific patterns. This is because the shared model is updated concurrently with the personalization of the local model, causing mutual influence between them and risking unintended privacy leakage. Unlike the methods above, PFedCLR decouples the update and personalization processes, only uploading the local model obtained from Step 1, before personalization occurs. The personalization is completed in Step 2, and the user-specific information is encapsulated in the low-rank matrices $\mathbf{A}_u\mathbf{B}_u$, which are never exposed to the server. This inherently preserves user privacy while maintaining model effectiveness.

\begin{figure*}[htbp]
\centering
\includegraphics[width=1.0\textwidth]{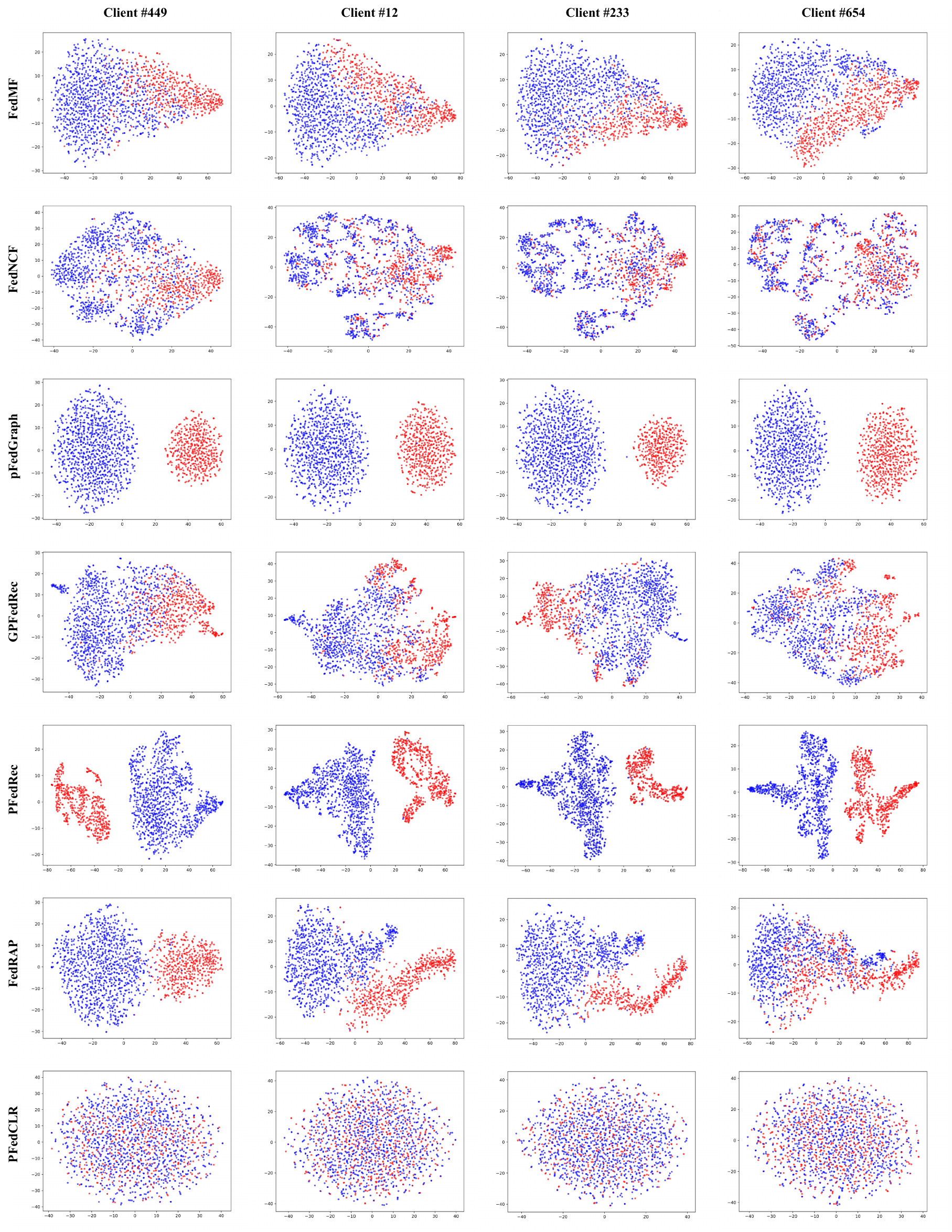} % Reduce the figure size so that it is slightly narrower than the column.
\caption{T-SNE visualizations of item embeddings uploaded to the server by different methods. Each row represents a different FR method, while each column represents a randomly sampled client. Items \textcolor{red}{interacted} with and \textcolor{blue}{not interacted} with by the client are indicated in \textcolor{red}{red} and \textcolor{blue}{blue}, respectively.}
\label{pic:vis_item}
\end{figure*}

II) Enhanced privacy protection. Our method can be integrated with Local Differential Privacy (LDP) to further enhance privacy. Table~\ref{tab:exp_frldp} presents the performance degradation of different methods with LDP on ML-100K. Notably, PFedCLR exhibits only a slight drop in both evaluation metrics, while still maintaining a clear performance advantage over other baselines. To balance privacy and utility, we adopt Laplace noise with a moderate strength of $\lambda=0.5$. Furthermore, we evaluate the robustness of PFedCLR under varying noise strengths. As shown in Table~\ref{tab:exp_moreldp}, PFedCLR remains robust when $\lambda \in [0.1, 0.5]$, and even under stronger privacy guarantees, \textit{e.g.}, $\lambda=1$, it still outperforms existing SOTA methods.

\begin{table}[ht]
\centering
\caption{Comparison of performance degradation with and without LDP for different method.}
\resizebox{\textwidth}{!}{
\begin{tabular}{llccccccc}
\toprule
\textbf{Metrics} & \textbf{Methods} & \textbf{FedMF} & \textbf{FedNCF} & \textbf{pFedGraph} & \textbf{GPFedRec} & \textbf{PFedRec} & \textbf{FedRAP} & \textbf{PFedCLR} \\
\midrule
\multirow{3}{*}{\textbf{HR@10}} & w/o LDP & 0.4846 & 0.4252 & 0.6204 & 0.6840 & 0.6702 & 0.8823 & 0.9989 \\
& w/ LDP & 0.4920 & 0.4199 & 0.6193 & 0.6448 & 0.6405 & 0.8441 & 0.9979 \\
& Degradation & - & $\downarrow$ 1.25\% & $\downarrow$ 0.18\% & $\downarrow$ 5.73\% & $\downarrow$ 4.43\% & $\downarrow$ 4.33\% & $\downarrow$ 0.10\% \\
\midrule
\multirow{3}{*}{\textbf{NDCG@10}} & w/o LDP & 0.2723 & 0.2290 & 0.4937 & 0.3982 & 0.3929 & 0.7980 & 0.9225 \\
& w/ LDP & 0.2665 & 0.2237 & 0.4857 & 0.3876 & 0.3624 & 0.7115 & 0.9215 \\
& Degradation & $\downarrow$ 2.13\% & $\downarrow$ 2.31\% & $\downarrow$ 1.62\% & $\downarrow$ 2.66\% & $\downarrow$ 7.76\% & $\downarrow$ 10.84\% & $\downarrow$ 0.11\% \\
\bottomrule
\end{tabular}
}
\label{tab:exp_frldp}
\end{table}

\begin{table}[ht]
\centering
\caption{Results of applying LDP into our method with different noise strength $\lambda$.}
\resizebox{\textwidth}{!}{
\begin{tabular}{llccccccc}
\toprule
\textbf{Dataset} & \textbf{Noise Strength} & \bm{$\lambda=0$} & \bm{$\lambda=0.1$} & \bm{$\lambda=0.2$} & \bm{$\lambda=0.3$} & \bm{$\lambda=0.4$} & \bm{$\lambda=0.5$} & \bm{$\lambda=1.0$} \\
\midrule
\multirow{2}{*}{\textbf{Filmtrust}} & HR@10 & 0.9102 & 0.9092 & 0.9122 & 0.9122 & 0.9072 & 0.9122 & 0.8882 \\
 & NDCG@10 & 0.7798 & 0.7797 & 0.7763 & 0.7776 & 0.7809 & 0.7741 & 0.6965 \\
\midrule
\multirow{2}{*}{\textbf{ML-100K}} & HR@10 & 0.9989 & 0.9979 & 0.9989 & 0.9989 & 0.9979 & 0.9979 & 0.9968 \\
 & NDCG@10 & 0.9225 & 0.9212 & 0.9211 & 0.9215 & 0.9177 & 0.9215 & 0.9059 \\
\midrule
\multirow{2}{*}{\textbf{ML-1M}} & HR@10 & 0.9603 & 0.9608 & 0.9609 & 0.9586 & 0.9594 & 0.9586 & 0.9593 \\
 & NDCG@10 & 0.8402 & 0.8402 & 0.8394 & 0.8385 & 0.8401 & 0.8379 & 0.8358 \\
\midrule
\multirow{2}{*}{\textbf{HetRec2011}} & HR@10 & 0.9522 & 0.9508 & 0.9498 & 0.9475 & 0.9446 & 0.9460  & 0.9319 \\
 & NDCG@10 & 0.8496 & 0.8505 & 0.8475 & 0.8473 & 0.8398 & 0.8405 & 0.8280 \\
\midrule
\multirow{2}{*}{\textbf{LastFm-2K}} & HR@10 & 0.7778 & 0.7770 & 0.7730 & 0.7738 & 0.7581 & 0.7526 & 0.7155 \\
 & NDCG@10 & 0.7164 & 0.7123 & 0.7102 & 0.7064 & 0.6914 & 0.6816 & 0.6509 \\
\bottomrule
\end{tabular}
}
\label{tab:exp_moreldp}
\end{table}

\section{Limitations}\label{app:limitation}
Our method introduces a local dual-function mechanism to simultaneously personalize global item embeddings and calibrate local user embedding. Considering client-side resource constraints, we implement this mechanism with low-rank decomposition, which incurs only a lightweight overhead relative to the backbone model. However, as a locally added component, it still imposes additional computation, which may hinder deployment on extremely resource-constrained clients. Moreover, PFedCLR adopts a simple weighted averaging scheme for global aggregation, \textit{i.e.}, FedAvg. In practice, this may limit the optimality of global updates, thereby constraining the overall recommendation performance. In future work, we plan to explore personalized and calibrated aggregation on the server side, aiming to alleviate client-side burdens.

%%%%%%%%%%%%%%%%%%%%%%%%%%%%%%%%%%%%%%%%%%%%%%%%%%%%%%%%%%%%

\end{document}